\documentclass[11pt]{article}
\usepackage{amssymb,amsmath, epsfig}
\include{epsf}

\usepackage[hypertex]{hyperref}
\usepackage[amsmath, hyperref, thmmarks]{ntheorem}
\theoremsymbol{}
\theorembodyfont{\slshape}
\theoremheaderfont{\normalfont\bfseries}
\theoremseparator{}

\theorembodyfont{\upshape}
\theoremsymbol{\ensuremath{\blacklozenge}}

\theoremstyle{nonumberplain}
\theoremheaderfont{\scshape}
\theorembodyfont{\normalfont}
\theoremsymbol{\ensuremath{\blacksquare}}

\newtheorem{proof}{Proof}
\qedsymbol{\ensuremath{_\blacksquare}}
\theoremclass{LaTeX}

\usepackage{amsfonts}
\usepackage{amscd}
\usepackage[all]{xy}
\advance \textheight by \topskip
\makeatletter

\@addtoreset{equation}{section}

\newcommand{\cg}[6]{
  \left(
  \begin{array}{cc|c}
  #1 & #3 & #5 \\
  #2 & #4 & #6
  \end{array}
  \right)
}
\makeatother

\newtheorem{lem}{Lemma}

\newtheorem{def-lem}[lem]{Definition/Lemma}

\newcommand{\be}{\begin{equation}}
\newcommand{\ee}{\end{equation}}
\newcommand{\beq}{\begin{equation}}
\newcommand{\eeq}{\end{equation}}
\newcommand{\bea}{\begin{eqnarray}}
\newcommand{\eea}{\end{eqnarray}}
\def\beqa{\begin{eqnarray}}
\def\eeqa{\end{eqnarray}}
\def\nn{\nonumber}
\newcommand{\Rl}{\mathbb{R}^3_\lambda}
\newcommand{\Rt}{\mathbb{R}^4_\theta}
\newcommand{\R}{\mathbb{R}}
\newcommand{\C}{\mathbb{C}}

\newcommand{\eqn}[1]{(\ref{#1})}
\newcommand{\del}{\partial}
\newcommand{\Tr}[1]{\:{\rm Tr}\,#1}
\newcommand{\dd}{{\mathrm d}}
\newcounter{lst}
\newenvironment{romanlist}
  {\begin{list}{\roman{lst})\ }{\usecounter{lst}
     \setlength{\labelwidth}{\leftmargin}\setlength{\labelsep}{0pt}}}
  {\end{list}}

  \def\appendix#1{\addtocounter{section}{1}\setcounter{equation}{0}
\renewcommand{\thesection}{\Alph{section}}
\section*{
\thesection\protect\indent \parbox[t]{11.715cm} {#1}}
\addcontentsline{toc}{section}{Appendix\thesection\ \ \ #1} }

\renewenvironment{thebibliography}[1]
         {\section*{References}\frenchspacing\small
          \begin{list}{[\arabic{enumi}]}
         {\usecounter{enumi}\parsep=2pt\topsep 0pt
         \settowidth{\labelwidth}{[#1]}
         \leftmargin=\labelwidth\advance\leftmargin\labelsep
         \rightmargin=0pt\itemsep=1pt\sloppy}}{\end{list}}
\title{Noncommutative field theories on $\mathbb{R}^3_\lambda$: \\
Towards UV/IR mixing freedom}
\author{Patrizia Vitale $^{a,b,c}$ and Jean-Christophe Wallet$^c$}

\begin{document}

\date{}
\maketitle
\vspace*{-1cm}

\begin{center}

\textit{$^a$Dipartimento di Scienze Fisiche
Universit\`a di Napoli Federico II}  \\
\textit{$^b$INFN, Sezione di Napoli, Via Cintia 80126 Napoli, Italy}\\
\smallskip
\textit{$^c$Laboratoire de Physique Th\'eorique, B\^at.\ 210\\
CNRS and Universit\'e Paris-Sud 11,  91405 Orsay Cedex, France}\\
\bigskip
 e-mail:
\texttt{patrizia.vitale@na.infn.it, jean-christophe.wallet@th.u-psud.fr}\\[1ex]

\end{center}

\begin{abstract}
We consider the noncommutative space $\Rl$, a deformation of the algebra of functions on $\R^3$ which yields a ``foliation" of $\mathbb{R}^3$ into fuzzy spheres. We first construct a natural matrix base adapted to $\Rl$.   We then apply this general framework to the one-loop study of a two-parameter family of real-valued scalar noncommutative field theories  with quartic polynomial interaction, which becomes a non-local matrix model when expressed in the above matrix base. The  kinetic operator involves a part related to dynamics on the fuzzy sphere  supplemented by a term reproducing  radial dynamics. We then compute the planar and non-planar 1-loop contributions to the 2-point correlation function. We find that these diagrams are both finite in the matrix base. We find no singularity of IR type, which signals very likely the absence of UV/IR mixing. We also consider the case of a kinetic operator with only the radial part. We find that the resulting theory is finite to all orders in perturbation expansion.
\bigskip

\end{abstract}

\pagebreak

\section{Introduction.}

Noncommutative Geometry (NCG) \cite{Connes1} provides a generalization of topology, differential geometry and index theory. The starting idea is to set-up a duality between spaces and associative algebras in a way to obtain an algebraic description of the structural properties of the space, in particular the topological, metric, differential properties \cite{Connes2}. Besides, many of the building blocks of fundamental physics fit well with concepts of NCG which may lead to a more accurate understanding of spacetime at short distance and/or possibly of what could be a quantum theory of gravity. For instance, NCG offers a possible way to treat the physical obstructions to the existence of a continuous space-time and commuting coordinates at the Planck scale \cite{Doplich1}. Once the noncommutative nature of space-time postulated, it is natural to consider  field theories on noncommutative manifolds, the so called Noncommutative Field Theories (NCFT). \par

The first prototypes of NCFT appeared in 1986 within String field theory \cite{witt1}. Field theories defined on the fuzzy sphere, a simple finite dimensional noncommutative geometry \cite{Hoppe, stratonowich, pepejoe2}, were introduced at the beginning of the 90's in \cite{fuzzy1}   and actively studied since then. See for example \cite{fuzzy2} for a review.  In1998 NCFT on the Moyal space, the simplest noncommutative geometry modeled on the phase-space of quantum mechanics, was shown to occur in effective regimes of String theory \cite{Schomerus}. This observation triggered a huge activity.  Noncommutative field theory of Moyal type was also shown to describe quite accurately   quantum Hall physics    \cite{qhe1}. For a review on NCFT on Moyal spaces see  \cite{douglas, szabo} and references therein.

The renormalization study of  NCQFT (Noncommutative Quantum Field Theory)  is in general difficult, a part from the case of finite noncommutative geometries, and is often complicated by the Ultraviolet/Infrared (UV/IR) mixing. This phenomenon occurs  for instance within the simplest noncommutative real-valued $\varphi^4$ model on the 4-dimensional Moyal space, as pointed out and analyzed in \cite{Minwalla}. The phenomenon persists in Moyal-noncommutative gauge models and represents one of the main open problems of Moyal-based field theory.  In  \cite{WAL1} noncommutative differential structures relevant to Moyal-noncommutative gauge theories  were studied, precisely to tackle such a problem.
A first solution for scalar field theory was  proposed in 2003 \cite{Grosse:2003aj}. It  amounts to modify the initial action with a harmonic oscillator term leading to a fully renormalisable NCQFT.
This is the so called Grosse-Wulkenhaar model. Various of its properties have been explored, among which classical and/or geometrical ones \cite{wal-golds},  the $2$-d fermionic extension \cite{V-T} as well as the generalization to gauge theory (matrix) models \cite{gauge-var1}, \cite{GWW}, \cite{gauge-var2},  \cite{wal-rev}. The Grosse-Wulkenhaar model has interesting properties such as the vanishing of the $\beta$-function to all orders \cite{beta} when the action is self-dual under the so-called Langmann-Szabo duality \cite{LSD}. Its $4$-d version is very likely to be non-perturbatively solvable, as shown in \cite{harald-raimar}. Besides, this model together with its gauge theory counterpart seems to be related to an interesting noncommutative structure, namely  a finite volume spectral triple \cite{harmonic-gw} whose relation to the Moyal (metric) geometries has been analyzed in \cite{homot-moyal, LVD}.

This  paper deals with a different kind of noncommutativity,  said of ``Lie algebra" type, because the $\star$-commutator of coordinate functions is not constant and reproduces the Lie bracket of classical Lie algebras. We shall follow  ref. \cite{selene} where many  $\star$-products were proposed, reproducing at the level of coordinate functions all three-dimensional Lie algebras, and in particular ref. \cite{Hammaa}, where the specific case of a $\mathfrak{su}(2)$ based star product giving rise to the noncommutative space $\R^3_\lambda$ was first introduced.
The purpose of this paper is twofold. The first goal is to set-up a general framework that can be used to study the quantum (i.e renormalisability) behaviour of matter NCFT as well as gauge NCFT \cite{undertak1-2} defined on $\Rl$. The second goal is to apply this framework to a class of natural scalar NCQFT on $\Rl$ in order to capture salient information related to its one-loop behaviour. There are important differences between $\Rl$ and the popular Moyal space. First of all,  the $\star$-commutator between the coordinates of $\Rl$ is no longer constant and  the relevant algebra of functions coding the $\Rl$ NCG is equipped with an associative but not {{translation-invariant}} product. Moreover,  the popular tracial property of the Moyal algebra \cite{pepejoe} does not hold true,  which complicates the treatment of the kinetic part of the action.
This difficulty can be handled by using a suitable matrix base. This is one of the results of the paper. We construct a natural matrix base adapted to $\Rl$ which can be obtained as a reduction of the matrix base of the Moyal space $\R^4_\theta$ \cite{pepejoe}\footnote{As we shall see in detail the starting matrix base in $\R^4_\theta$ is actually a slight modification of the Moyal matrix base which was introduced in \cite{pepejoe}. We shall use a matrix base adapted to the Wick-Voros product.}.
We then consider a family of real-valued scalar actions on $\Rl$ with quartic polynomial interactions. The family of kinetic operators, indexed by two real parameters,  involves a natural Laplacian-type operator which contains the square of the angular momentum and an additional term  related (but not equal) to the Casimir operator of $\mathfrak{su}(2)$, which is necessary in order to generate some radial dynamics. When re-expressed in the natural matrix base, the action is the one for a non-local matrix model with {\it{diagonal}} interaction term\footnote{In physical language, the natural base is nothing but the interaction base.} and {\it{non-diagonal}} kinetic operator being of Jacobi type. The action can be split as an infinite sum of scalar actions defined on the successive fuzzy spheres that ``foliate'' the noncommutative space $\Rl$.  The additional term in the Laplacian encodes radial dynamics.

Another matrix base, largely used in the literature on the fuzzy sphere,  is built from the symbols of the fuzzy spherical harmonic operators and   related to the natural base of $\Rl$. This  leaves {\it{diagonal}} the kinetic operator with however a complicated interaction term\footnote{This other base is physically the so-called propagation base. See the previous footnote. We will use this terminology when appropriate in the paper. }. Upon diagonalizing the kinetic term, the computation of the propagator in the  natural  base can then be performed  and we end up with   a tractable expression together with a purely diagonal interaction term.
We then compute the planar and non-planar 1-loop contributions to the 2-point correlation function. We find that they are both finite in the natural matrix base. The computation of the corresponding amplitudes in the propagation base, when only the angular momentum part of the kinetic operator is involved, shows consistency with previous work on the fuzzy sphere \cite{vaidya01,chu01}. We find no IR singularities, which signals very likely the absence of the UV/IR mixing phenomenon at the perturbative level.

We also consider the limit situation where the Laplacian is only given by the term related to the Casimir operator. This leads to a big simplification for the action and  for the general power counting of the ribbon diagrams of arbitrary orders. We find that the resulting theory is finite to all orders in perturbation.\par

The paper is organized as follows. In section \ref{sect2} we summarize the general properties of the noncommutative $\Rl$ that will be used in this paper together with some features related to the Wick-Voros product. In  section \ref{matrixbasesection} we construct a natural matrix base adapted to $\Rl$. In  section \ref{theactions} we construct a family of real-valued scalar actions on $\Rl$ with quartic polynomial interactions. The relationship with the base built from the fuzzy spherical harmonics is also introduced. The subsection \ref{subsect4.3} involves the computation of the propagator expressed in the natural matrix base for $\Rl$, which is rather easily carried out, once a suitable combination of the change of base of subsection \ref{subsect4.2} with properties of fuzzy spherical harmonics is done. Nevertheless, we find interesting to provide in the appendix the general computation of the propagator that takes advantage of the Jacobi nature of the kinetic operator. This is based on the determination of a suitable family of orthogonal polynomials that gives rise to diagonalization. We find in the present case that the relevant orthogonal polynomials are the dual Hahn polynomials, the counterpart of the Meixner polynomials underlying the computation of the Grosse-Wulkenhaar propagator. This, as a by-product, provides explicit relations between fuzzy spherical harmonics, Wigner $3j$-symbols and dual Hahn polynomials.   In section \ref{sect5} we compute and discuss the planar and non-planar 1-loop contribution to the 2-point correlation function, respectively in the subsections \ref{planarsection} and \ref{nonplanarsection}. In the subsection \ref{finitemodels} we consider the  limit case of a kinetic operator with no angular momentum term,  for which we find finitude to all orders in perturbation. We finally summarize the results and conclude.\par

\section{The noncommutative space $\Rl$}\label{sect2}
The noncommutative space $\Rl$ has been first introduced in \cite{Hammaa}. A generalization has been studied in \cite{selene}.  It is  a subalgebra of $\Rt$, the noncommutative algebra of functions on $\R^4\simeq \C^2$ endowed with the Wick-Voros product \cite{Wick-Voros}
\be
\phi\star \psi\, (z_a,\bar z_a)= \phi(z,\bar z) \exp(\theta \overleftarrow\del_{z_a}\overrightarrow\del_{\bar z_a}) \psi(z,\bar z), \,\,\,\, a=1,2 \label{Wick-Vorosprod}
\ee
This is an asymptotic expansion; a proper definition, based on the dequantization map  associated to normal ordered quantization,will be given   in sect. \ref{matbasesubsect1} where it is actually needed to introduce a matrix base.  For coordinate functions we have the $\star$-commutator
\be
[z_a,\bar z_b]_\star= \theta \delta_{ab}
\ee
with $\theta$ a constant, real parameter. Resorting to real coordinates $q_a= z_a+\bar z_a, p_a= i(z_a-\bar z_a),\,a=1,2$  we recover the usual $\star$-commutator of (two copies of) the Moyal plane, the two products differing by symmetric terms.\footnote{The Moyal and Wick-Voros  algebras are isomorphic \cite{equivalence}, there is however a  debate on the physical meaning of the equivalence between them, which we will not address here (see for example \cite{GLV08,GLV09, bala,  scholz}).}
The crucial step  to obtain star products on $\mathcal{F}(\mathbb{R}^3)$, hence to deform $\mathcal{F}(\mathbb{R}^3)$ into a noncommutative algebra, is to identify
$\mathbb{R}^3$ with the dual, $\mathfrak{g}^*$,  of some chosen three dimensional Lie
algebra $\mathfrak{g}$. We choose here to work with the $\mathfrak{su}(2)$ Lie algebra, because of the connection with other results already present in the literature (for example the fuzzy sphere) but other choices can be made. This identification induces on $\mathcal{F}(\mathbb{R}^3)$ the
Kirillov Poisson bracket, which, for coordinate functions reads
\be
\{x_i,x_j\}= c_{ij}^k x_k \label{Kirillov}
\ee
with $i=1,..,3$ and  $c_{ij}^k $ the structure constants of $\mathfrak{su}(2)$.  On the other
hand, it is well known that this  (Poisson) Lie algebra may be regarded as a subalgebra of the
symplectic algebra $\mathfrak{sp}(4)$, which is classically realized  as the Poisson
 algebra of quadratic functions on $\mathbb{R}^4$ ($\mathbb{C}^2$ with
our choices) with canonical Poisson bracket
\be
\{z^a, \bar z^b\}= i .
\ee
Indeed it is possible to find quadratic functions
\be \pi^*(x_i)= \pi^*(x_i)(z^a,\bar z^a)
\ee
 which obey \eqn{Kirillov}.We have indicated with $\pi^*$ the  pull-back map $\pi^*:\mathcal{F}(\R^3)\rightarrow \mathcal{F}(\R^4)$.
This is nothing but the classical counterpart of the Jordan-Schwinger map realization of
Lie algebra generators in terms of creation and annihilation operators \cite{MMVZ94}. Then
one can show  that this Poisson subalgebra is also a Wick-Voros  subalgebra,
that is
\be
\pi^*(x_i)(z^a,\bar z^a)\star \pi^*(x_j)(z^a,\bar z^a) -\pi^*(x_j)(z^a,\bar z^a)\star \pi^*(x_i)(z^a,\bar
z^a)= \lambda c_{ij}^k \pi^*(x_k)(z^a,\bar z^a) \label{starx}
\ee
where the noncommutative parameter $\lambda$  shall be adjusted according to  the physical dimension of the coordinate functions $x_i$.
We shall indicate with $\mathbb{R}^3_\lambda$ the  noncommutative
algebra $(\mathcal{F}(\R^3), \star)$.  Eq. \eqn{starx}  induces a star product on polynomial functions on
$\mathbb{R}^3$ generated by the coordinate functions $x_i$, which may be expressed in
closed form in terms of differential operators on $\mathbb{R}^3$.   Here we will consider quadratic
realizations of the kind
\be
\pi^*(x_\mu)=\frac{\lambda}{\theta} \bar z^a e_\mu^{ab} z^b, \;\;\; \mu=0, ..,3 \label{xmu}
\ee
with $\lambda$ a constant, real parameter of length dimension equal to one;   $e_i= \frac{1}{2}\sigma_i, \: i=1,..,3$ are the $SU(2)$ generators in the defining representation with  $\sigma_i$
 the Pauli matrices, while $e_0=\frac{1}{2} \mathbf{1}$.  We  shall omit the pull-back map from now on,  unless necessary.
 Notice that
\be
x_0=\frac{\lambda}{2\theta}\bar z_a z_a
\ee
commutes with $x_i$ so that we can alternatively define $\R^3_\lambda$ as the commutant of $x_0$; $x_0$ generates the center of the algebra. We also have
\be
x_0^2= \sum_i x_i^2. \label{x02}
\ee
It is easily verified that the induced  $\star$-product reads
\be
\phi\star \psi \,(x)= \exp\left[\frac{\lambda}{2}\left(\delta_{ij} x_0+ i \epsilon_{ij}^k x_k \right)\frac{\del}{\del u_i}\frac{\del}{\del v_j}\right] \phi(u) \psi(v)|_{u=v=x} \label{starsu2}
\ee
which implies, for coordinate functions
\beqa
x_i\star x_j &=& x_i x_j+ \frac{\lambda}{2} \left(x_0 \delta_{ij} + i \epsilon^{ij}_k x^k\right) \label{xistarxj}\\
x_0\star x_i &=& x_i\star x_0 = x_0 x_i + \frac{\lambda}{2} x_i\\
x_0\star x_0&=& = x_0(x_0+\frac{\lambda}{2}) = \sum_i x_i\star x_i - \lambda x_0  \label{x0*2}
 \eeqa
 where Eq.  \eqn{x02} has been used, together with the equality $\sum_i x_i\star x_i= \sum_i x_i^2 + 3/2 \; \lambda \,x_0\, $ descending from Eq. \eqn{xistarxj}. The product is associative, since it is nothing but the Wick-Voros product expressed in different variables.
 As for the $\star$ commutator we have
  \be
 [x_i,x_j]_\star=i \lambda \epsilon_{ij}^k x_k . \label{commsu2}
 \ee
On introducing  the parameter $\kappa=\lambda/\theta$, the commutative limit is achieved with $\lambda, \theta \rightarrow 0$. $\kappa={\rm const}$.

We have thus realized the announced isomorphism between the algebra of linear functions on $\R^3\simeq \mathfrak{su}(2)^*$ endowed with the $\star$ commutator \eqn{commsu2}  and the $\mathfrak{su}(2)$ Lie algebra. Thus,  the algebra $\mathbb{R}^3_\lambda$ can be defined as $\mathbb{R}^3_\lambda=\mathbb{C}[x_\mu]/{\cal{I}}_{{\cal{R}}_1,{\cal{R}}_2 }$, i.e the quotient of the free algebra generated by the coordinate functions $(x_i)_{i=1,2,3}, x_0,$ by the two-sided ideal generated by the relation ${\cal{R}}_1: [x_i,x_j]_\star=i\lambda\epsilon_{ijk}x_k$, together with ${\cal{R}}_2: x_0\star x_0 +\lambda x_0 =\sum_i x_i\star x_i$. Notice that, because of the presence of $x_0$ $\mathbb{R}^3_{\lambda\ne 0}$ {\it is not} isomorphic to  $ {\cal{U}}(\mathfrak{su}(2))$.

For a comparison with the Moyal induced product we refer to \cite{selene}.

\section{\texorpdfstring{The matrix base }{Matrix base }}\label{matrixbasesection}

The matrix base we  shall define  for $\R^3_\lambda$ is obtained through  a suitable reduction of the matrix base of the Wick-Voros algebra $\R^4_\theta$, which was introduced in \cite{discofuzzy}. The latter is in turn a slight modification of the well known matrix base for the Moyal algebra defined in \cite{pepejoe} by  Gracia-Bond\`{\i}a and Varilly.

\subsection{The matrix base for the Wick-Voros $\R_\theta^4$}\label{matbasesubsect1}
 Let us first review the matrix base adapted to the Wick-Voros algebra $\R_\theta^4$. Our convention, all over the paper, will be to use hatted letters to indicate
operators and un-hatted ones to indicate their noncommutative symbols.

It is well known that the Wick-Voros product is introduced through a weighted quantization map which, in two dimensions, associates to functions on the complex plane the operator
\be
\hat \phi=\hat {\mathcal{W}}_V(\phi)=\frac{1}{(2\pi)^2}\int\dd^2 z \,\hat \Omega (z, \bar z)
\phi(z, \bar z) \label{Weylmap}
\ee
where
\be
\hat \Omega (z, \bar z)= \int \dd^2 \eta\, e^{{ -( \eta \bar z-\bar\eta z )}} e^{
\theta\eta a^\dag} e^{-\theta\bar\eta a }
\ee
is the so called quantizer and $a, a^\dag$ are the usual (configuration space) creation
and annihilation operators, with commutation relations
\be
[a, a^\dag]=\theta. \label{commz}
\ee
A word of caution is in order, concerning the domain and the range of the weighted Weyl map in Eq. \eqn{Weylmap}.  While the standard Weyl map maps Schwarzian functions into Hilbert
Schmidt operators, for the weighted Weyl map \eqn{Weylmap} this is not always the case. An  exhaustive analysis is lacking in the literature, up to our knowledge. Explicit counterexamples are discussed in \cite{discofuzzy}.

The inverse map which is the analogue of the Wigner map  is represented by:
\be
\phi(z,\bar z) = \mathcal{W}_V^{-1}(\hat \phi)= \langle z|\hat \phi|z\rangle \label{wigner}
\ee
with $|z\rangle$ the {\it coherent states} defined by $a|z\rangle=z|z\rangle$. Notice that, differently from the Weyl-Wigner-Moyal case, the quantizer and dequantizer operators do not coincide,
meaning that this
quantization/dequantization procedure  is not self-dual  (see \cite{pepejoe,convolodya} for
details).

The Wick-Voros product, whose asymptotic form has been already given in \eqn{Wick-Vorosprod}, is then defined as
\be
\phi\star \psi := \mathcal{W}_V^{-1}\left(\hat{\mathcal{W}}_V(\phi)\hat {\mathcal{W}}_V(\psi)\right)= \langle z|\hat \phi\,\hat \psi|z\rangle
\label{Wick-Voros}
\ee
It can be seen that, for analytic functions, a very convenient way to reformulate the quantization map \eqn{Weylmap} is to consider their analytic expansion
\be
\phi(\bar z,z)=\sum_ {pq}\tilde\phi_{pq} \bar z^p z^q\, ,\;\; p,q\in \mathbb{N}  \label{anexp}
\ee
with $\tilde\phi_{pq}\in \C$.
The quantization map \eqn{Weylmap} will then produce   the normal ordered operator
\be
\hat \phi =\hat {\mathcal{W}}_V(\phi)=\sum_ {pq}\tilde\phi_{pq}  a^{\dag p} a^q \label{normord}
\ee
We will therefore assume analyticity all over the paper and use \eqn{anexp}, \eqn{normord}.

This construction may be easily extended to $\R^{2n}$. We will consider $n=2$ from now on.

Each $a_a$, $a=1,2$, acts on ${\cal{H}}_0\cong \ell^2(\mathbb{N})$, a copy of the Hilbert space of the one dimensional harmonic oscillator with canonical orthonormal base $(|n\rangle)_{n\in\mathbb{N}}$. We set
\be
|N\rangle=|n_1,n_2\rangle:=|n_1\rangle\otimes|n_2\rangle \label{numberbase}
\ee
 the canonical orthonormal base for ${\cal{H}}={\cal{H}}_0\otimes{\cal{H}}_0$, also called in the physics literature the number base. The action of the
$a_a,a_a^\dag$'s on ${\cal{H}}$ is given by
\beqa
a_1|n_1,n_2\rangle&=&{\sqrt{\theta}}{\sqrt{n_1}}|n_1-1,n_2\rangle,\ {{a_1^\dag}}|n\rangle={\sqrt{\theta}}{\sqrt{n_1+1}}|n_1+1,n_2\rangle,\nn\\
 a_2|n_1,n_2\rangle&=&{\sqrt{\theta}}{\sqrt{n_2}}|n_1,n_2-1\rangle,\ {{a_2^\dag}}|n\rangle={\sqrt{\theta}}{\sqrt{n_2+1}}|n_1,n_2+1\rangle\label{oscharmonic}.
\eeqa
For further use, we also define for any $a=1,2$ the number operators $N_a={{a}}_a^\dag a_a$ satisfying
\be
N_a|n\rangle=\theta n|n\rangle\;\;\;\forall |n\rangle\in{\cal{H}}_0.
\ee
To functions on $\R^4$ we associate via the quantization map \eqn{normord} normal ordered operators.
On using the number base \eqn{numberbase} together with \eqn{oscharmonic} we may rewrite \eqn{normord} as
\be
\hat\phi=\sum_{P,Q\in\mathbb{N}^2}\phi_{PQ}|P\rangle \langle Q| \;\;\; \phi_{MN}\in \C
\ee
with $\tilde\phi_{PQ}, \phi_{LK}$ related by a change of base.
 We have indeed
\begin{equation}
|P\rangle={{{{a}}_1^{\dag p_1}{{a}}_2^{^\dag p_2}}\over{[P!\theta^{|P|}]^{1/2} }}|0\rangle,\;\;  \forall P=(p_1,p_2)\in\mathbb{N}^2\label{matrixbase0},
\end{equation}
with $P!:=p_1!p_2!,\ |P|:=p_1+p_2,$ and $|0\rangle=|0,0\rangle$  a Fock vacuum state. This implies
\be
\langle z_1,z_2|P\rangle=\overline{\langle P|z_1,z_2\rangle}=e^{-\frac{ \bar z_1 z_1+\bar z_2 z_2}{2\theta}}\frac{\bar z_1^{p_1} \bar z_2^{p_2}}{P! \theta^{|P|}}. \label{zn}
\ee
Thus
\be
\phi_{LK}=\sum_{q_1=0}^{min\,(l_1,k_1)}  \sum_{q_2=0}^{min\,(l_2,k_2)} \tilde \phi_{l_2-q_2,k_2-q_2} \frac{\sqrt{L! K! \theta^{|L|+|K|}}}{\theta^{|Q|} Q!}.\label{changeofbase}
\ee
On applying the dequantization map we obtain a function in the noncommutative Wick-Voros algebra
\be
\phi (z,\bar z)=\sum_{P Q} \phi_{PQ } f_{P Q} (z,\bar z)
\ee
with
\be
f_{PQ}(z,\bar z)= \langle z_1,z_2|\hat f_{PQ}|z_1,z_2\rangle=
\frac{e^{-\frac{ \bar z_1 z_1+\bar z_2 z_2}{\theta}}}{\sqrt{P!Q!\theta^{|P+Q|}}} \bar z_1^{p_1}\bar z_2^{p_2}  z_1^{q_1} z_2^{q_2}  \label{wigner functions}
\ee
and we have introduced the notation $\hat f_{PQ}= |P\rangle\langle Q|$. Here we have used \eqn{zn}.

 The base operators $\hat f_{PQ}$
 fulfill the following fusion rule, approximation of the identity  and trace property  respectively given by:
\begin{equation}
\hat f_{MN}\hat f_{PQ}=\delta_{NP}\hat f_{MQ},\ \mathbf{1}=\sum_{M\in\mathbb{N}^2}|M\rangle\langle M|,\ Tr(\hat f_{MN})=\delta_{MN}, \forall M,N,P,Q\in\mathbb{N}^2\label{properties1}.
\end{equation}
These properties descend to the symbol functions of the base $\hat f_{MN}$, defining
 an orthogonal matrix base for $\R^4_\theta$ with a simple rule for the star product
\be
f_{MN}\star f_{PQ}(z,\bar z)= \langle z_1,z_2| \hat f_{MN}\hat f_{PQ}|z_1,z_2\rangle =\delta_{NP} f_{MQ}(z,\bar z)
\ee
and analogous expressions for the other relations. In particular, on using the decomposition of the identity in terms of coherent states
\be
\mathbf{1}=\frac{1}{(\pi\theta)^2}\int d^2z_1 d^2 z_2 |z_1,z_2\rangle\langle z_1,z_2|
\ee
and the last identity in \eqn{properties1}, we arrive at
\be
\int \dd^2 z_1 \dd^2 z_2\, f_{PQ}(z,\bar z)= (\pi \theta)^2 \delta_{PQ}.
\ee
The same result can be obtained by direct calculation on using \eqn{Wick-Vorosprod} and \eqn{wigner functions}.

The existence of an orthogonal matrix base for $\R^4_\theta$ allows to rewrite the Wick-Voros product in \eqn{Wick-Vorosprod} as a matrix multiplication. To this we introduce the notation $\Phi:=\{\phi_{PK}\}_{P,K\in \mathbb{N}^2}$ for the infinite matrix whose entries are the fields components. We have then
\be
\phi\star\psi\,(z,\bar z)= \sum_{PQ}\sum_{LK} \phi_{PQ}\psi_{LK} (f_{PQ}\star f_{LK})(z, \bar z)= \sum_{PK}(\Phi\cdot \Psi)_{PK} f_{PK}(z, \bar z)
\ee
with $( \;\cdot \;)$ the matrix product.
As already noticed, this is a slight modification, adapted to the Wick-Voros product, of the matrix base defined in \cite{pepejoe} for the Moyal algebra. It was introduced in \cite{discofuzzy}  and already used in \cite{conadrian} in the context of renormalizable scalar field theories on the Wick-Voros plane. In the context of quantum field theories approximated with fuzzy geometries the Wick-Voros base has been recently used in \cite{spisso}.

\subsection{The matrix base of $\mathbb{R}^3_\lambda$}\label{subsec.3.2}
In order to obtain a matrix base in three dimensions, compatible with the product
\eqn{starsu2}, we resort to the Schwinger-Jordan realization of the $\mathfrak{su}(2)$
Lie algebra in terms of creation and annihilation operators, which was given in \eqn{xmu}.  The derivation is identical to the one performed in \cite{rosa12} except for the fact that the starting point, the matrix base on $\R^4_\theta$ is here the Wick-Voros one.

It is known that  ${\cal{H}}={\cal{H}}_0\otimes{\cal{H}}_0$ admits the natural decomposition ${\cal{H}}=\bigoplus_{ j\in{{\mathbb{N} }\over{2}}  }{\cal{V}}_j$ where
\begin{equation}
{\cal{V}}_j=span\{ |j,m\rangle\}_{-j\le m\le j},\ |j,m\rangle:=|j+m\rangle\otimes| j-m\rangle\label{irrepssu2}
\end{equation}
is the linear space carrying the irreducible representation of $SU(2)$ with dimension $2j+1$ and for any $j\in{{\mathbb{N} }\over{2}} $, the system $\{ |j,m\rangle\}_{-j\le m\le j}$ is orthonormal. From this, it can be realized that another natural base for $\mathbb{R}^4_\theta$ is provided by
\begin{equation}
\{\hat v^{j\tilde\jmath}_{m\tilde m}:=|j,m\rangle\langle \tilde\jmath,\tilde m|\},\ j,\tilde\jmath\in{{\mathbb{N}}\over{2}},\ -j\le m\le j\, , -\tilde\jmath\le \tilde m\le \tilde\jmath \label{2sphericbase}\, .
\end{equation}
The two bases are related as follows.
We observe that the
eigenvalues of the number operators $\hat N_1=a^\dag_1 a_1$, $\hat N_2=a^\dag_2 a_2$, say  $p_1,
p_2$, are related to the eigenvalues of $\hat{\mathbf{X}}^2, \hat X_3$, respectively $j(j+1)$ and $m$,   by
\be
p_1+p_2=2j\;\;\; p_1-p_2=2m
\ee
with $p_i\in \mathbb{N}$, $j\in \mathbb{N}/2$, $-j\le m \le j$,
so to have
\be
|p_1 p_2\rangle=|j+m, j-m\rangle\equiv |j m\rangle=
\frac{(a_1^\dag)^{j+m}(a_2^\dag)^{j-m}}{\sqrt{(j+m)!(j-m)!\theta^{2j}}} |00\rangle
\ee
where $\hat X_i, i=1,..,3$ are the standard angular momentum operators  representing the $\mathfrak{su}(2)$ Lie algebra in terms of selfadjoint operators on the Hilbert space $ \mathcal{V}_j$ spanned by $|j,m\rangle$.
Thus we may identify
\be
\hat f_{MN}\rightarrow \hat v^{j \tilde\jmath}_{m  \tilde m},
\ee
and, for their symbols
\be
f_{MN}(z,\bar z)\rightarrow  v^{j \tilde\jmath}_{m  \tilde m} (z,\bar z)=\langle z_1,z_2 |\hat v^{j \tilde\jmath}_{m  \tilde m}| z_1,z_2\rangle
\ee
so to have, for $\phi\in\R^4_\theta$
\be
\phi(z_a,\bar z_a)= \sum_{{j}, {\tilde\jmath}\in \mathbb{N}/2}\sum_{{m}=-{j}}^{j} \sum_{ {\tilde m}=-{\tilde\jmath}}^{\tilde\jmath}
\phi^{{j} {\tilde\jmath}}_{{m} {\tilde m}} v^{{j}{\tilde\jmath}}_{{m}{\tilde m}}(z_a,\bar z_a)
\ee
We further observe that,  for  $\phi (z,\bar z)$ to be in the subalgebra
$\mathbb{R}^3_\lambda$ we must impose $j=\tilde\jmath$. To this it suffices to
compute
\be
x_0 \star v^{j\tilde\jmath}_{m\tilde m} (z,\bar z)-v^{j\tilde\jmath}_{m\tilde m}\star
x_0 (z,\bar z)=\lambda(j-\tilde\jmath) v^{j\tilde\jmath}_{m\tilde m}
\ee
with $x_0(z,\bar z)=\lambda/(2\theta) \bar z_a z_a $ and the $\star$ product defined in \eqn{starsu2}. Then we recall  that $\mathbb{R}_\lambda^3$  may be alternatively defined as
the $\star$-commutant of $x_0$. This imposes
\be
j=\tilde \jmath
\ee
We have then
\be
\phi(x_i, x_0)=\sum_{j}\sum_{m,\tilde m=-j}^j \phi^j_{m\tilde m} v^j_{m\tilde m}
\label{phixi}
\ee
with
\be
v^j_{m\tilde m}:= v^{jj}_{m\tilde m}= e^{-\frac{\bar z_a z_a}{\theta}}\frac{\bar z_1^{j+m}
z_1^{j+\tilde m} \bar z_2^{j-m}
z_2^{j-\tilde m}}{\sqrt{(j+m)!(j-m)! (j+\tilde m)!(j-\tilde m)! \theta^{4j} }} \label{matrixbase}
\ee
and we recall  its expression in terms of the dequantization map
\be
v^j_{m\tilde m}(z,\bar z)=\langle z_1,z_2|j\,m\rangle\langle j\,\tilde m|z_1,z_2\rangle . \label{spinwignerfunctions}
\ee
As for the normalization we have
\be
\int \dd^2 z_1\dd^2 z_2 \;v^j_{m \tilde m} (z, \bar z)=  \pi^2 \theta^2 \delta_{m
\tilde m}.
\ee
Let us notice that the base elements $v^j_{m\tilde m}(z,\bar z)$ can be reexpressed solely in terms of the coordinate functions $x_i, x_0$ although the expression is not unique. A possible choice is
\be
v^j_{m\tilde m}(x_i,x_0)= \frac{ e^{-2\frac{ x_0}{\lambda}}}{ \lambda^{2j}}  \frac{(x_0+x_3)^{j+m}
(x_0-x_3)^{j-\tilde m}\; (x_1-i x_2)^{\tilde m -m}
}{\sqrt{(j+m)!(j-m)! (j+\tilde m)!(j-\tilde m)! }}.
\ee
The star product acquires the simple form
\be
v^j_{m\tilde m}\star v^{\tilde\jmath}_{n \tilde n}=\delta^{j \tilde\jmath}\delta_{\tilde m
n}v^j_{m \tilde n} \label{simplform}
\ee
which implies the orthogonality property
\be
\int
v^j_{m\tilde m}\star v^{\tilde\jmath}_{n \tilde n}=\pi^2 \theta^2\delta^{j \tilde\jmath}\delta_{\tilde m
n}\delta_{m \tilde n}. \label{ortov}
\ee
These properties may be either directly verified or derived from the dequantization map starting from the operator relations
\begin{equation}
\hat v^{j_1}_{m_1,m_2}\hat v^{j_2}_{n_1,n_2} =\delta_{j_1j_2}\delta_{m_2n_1}\hat v^{j_1}_{m_1,n_2},\ {{(\hat v^{j}_{m_1,m_2})^\dag}}=\hat v^{j}_{m_2,m_1}
\label{properties2a},
\end{equation}
\begin{equation}
\langle \hat v^{j_1}_{m_1,m_2},\hat v^{j_2}_{n_1,n_2} \rangle=\delta_{j_1j_2}\delta_{m_1n_1}\delta_{m_2n_2}\;\;\; \mathbf{1}=\sum_{j\in {{\mathbb{N}}\over{2}} }\sum_{m=-j}^j \hat v^j_{mm},\ \Tr(\hat v^j_{m_1,m_2})=\delta_{m_1m_2}\label{properties2b}
\end{equation}
where $\langle \hat v^{j_1}_{m_1,m_2},\hat v^{j_2}_{n_1,n_2}\rangle :=\Tr (\hat v^{j_1}_{m_1,m_2})^\dag \hat v^{j_2}_{n_1,n_2}$ is the scalar product.

Notice that, for any ${j\in { {{ {\mathbb{N}} }\over{2} }}}$, the set $\{\hat v^j_{m_1,m_2}\},\ -j\le m_1,m_2\le j$ of $(2j+1)^2$ linear maps $v^j_{m_1,m_2}:{\cal{V}}^j\to{\cal{V}}^j$ simply describes the canonical base of the algebra of endomorphisms of ${\cal{V}}^j$, orthonormal with respect to the scalar product introduced above.
 From this it follows that the direct sum decomposition
\begin{equation}
\mathbb{R}^3_\lambda\simeq\bigoplus_{j\in { {{ {\mathbb{N}} }\over{2} }}}End({\cal{V}}^j) \simeq\bigoplus_{j\in { {{ {\mathbb{N}} }\over{2} }}}\mathbb{S}^j\label{sumfuzzy}
\end{equation}
holds true, where $End({\cal{V}}^j)$ denotes the algebra of endomorphisms of ${\cal{V}}^j$, $\forall {j\in { {{ {\mathbb{N}} }\over{2} }}}$, which actually describe the so-called fuzzy spheres of different radii, ${\mathbb{S}}^j$.

The star product in $\mathbb{R}_\lambda^3$ becomes a block-diagonal matrix product
\beqa
\phi\star \psi (x_i,x_0)&=&\sum \phi^{j}_{m_1\tilde m_1} \psi^{j}_{m_2\tilde m_2}
v^{j}_{m_1\tilde m_1} \star v^{j}_{m_2\tilde m_2} \; =\sum \phi^{j}_{m_1\tilde m_1}
\psi^{j}_{m_2\tilde m_2} v^{j}_{m_1\tilde m_2} \delta_{\tilde m_1 m_2}\nn\\
&=& \sum_{j,m_1, \tilde m_2} (\Phi^j\cdot \Psi^j)_{m_1 \tilde m_2} v^j_{m_1 \tilde m_2}
\label{starmatrix}
\eeqa
where the infinite matrix $\Phi$ gets rearranged into a block-diagonal form, each block being the $(2j+1)\times (2j+1)$ matrix   $\Phi^j=\{\phi^j_{mn}\}, \, -j\le m,n\le j$.
The integral may be defined through the pullback to $\R^4_\theta$
\be
\int_{\mathbb{R}^3_\lambda} \phi
:=\frac{\kappa^3}{\pi^2\theta^2} \int_{\mathbb{R}^4_\theta}\pi^\star(\phi)=  \kappa^3 \sum_j \Tr_j \Phi^j
\ee
with $\Tr_j$ the trace in the $(2j+1)\times(2j+1)$ subspace\footnote{
If we were to perform our analysis in the coordinate base, without recurring to the matrix base, we  should use a differential calculus adapted to   $\R^3_\lambda$ as the one introduced in  \cite{diffcalc} .}.
We have also
\be
\int_{\mathbb{R}^3_\lambda} \phi \star \psi
:=\frac{\kappa^3}{\pi^2\theta^2} \int_{\mathbb{R}^4_\theta}\pi^\star(\phi)\star \pi^*(\psi)= \kappa^3 \sum_j \Tr_j \Phi^j \Psi^j.
\ee

\section{The scalar actions}\label{theactions}
In this section we consider a family of scalar field theories  on $\R^3_\lambda$
indexed by two real parameters $\alpha,\beta$. We assume the fields $\phi \in R^3_\lambda$ to be real. Upon rewriting the action in the matrix base we perform one loop calculations and discuss the divergences. Some comments on the renormalization of the theory are given at the end.
\subsection{General properties}\label{subsect4.1}
Let
\be
S[\phi]=\int  \phi\star(\Delta+\mu^2) \phi + \frac{g}{4!}\phi\star\phi\star\phi\star\phi \label{action}
\ee
where $\Delta$ is the Laplacian defined as
\be
\Delta \phi = \alpha \sum_i D_i^2\phi + \frac{\beta}{\kappa^4}x_0\star x_0\star\phi\label{lapl}
\ee
and
$D_i=\kappa^{-2}[x_i, \; \cdot\;]_
\star, \; i=1,..,3$ are inner derivations of $\R^3_\lambda$. The mass dimensions are $[\phi]=\frac{1}{2},\, [g]=1, \, [D_i]=1$. $\alpha$ and $\beta$ are dimensionless parameters.

The second term in the Laplacian has been added in order to introduce radial dynamics. From \eqn{starsu2} we have indeed
\be
[x_i,\phi]_\star= - i\lambda \epsilon_{ijk}x_j \del_k \phi
\ee
so that the first term,  that is $[x_i,[x_i,\phi]_\star]\star$ can only reproduce tangent   dynamics on fuzzy spheres; this is indeed the Laplacian usually introduced for quantum field theories on the fuzzy sphere.  Whereas
\be
x_0\star \phi = x_0 \phi + \frac{\lambda}{2} x_i \del_i \phi \label{radialgenerat}
\ee
 contains the dilation operator in the radial direction.

Therefore, the highest derivative term of the Laplacian defined in \eqn{lapl} can be made into the ordinary Laplacian on $\R^3$ multiplied by $x_0^2$,  for the parameters $\alpha$ and $\beta$ appropriately  chosen.
We have indeed
\beqa
\sum_i [x_i,[x_i,\phi]_\star]_\star&=& \lambda^2 \bigl[x^i\del_i (x^j\del_j\phi+x^i\del_i \phi)\bigr] -\lambda^2 x_0^2\del^2\phi
\\
 x_0\star x_0\star\phi +\frac{\lambda}{2} x_0\star\phi &=& \frac{\lambda^2}{4}\bigr(x^i\del_i (x^j\del_j\phi)+ x^i\del_i\phi\bigr) +\lambda x_0(x^i\del_i\phi+\phi) +x_0^2\phi  \label{radialact}
\eeqa
where,  in order to have homogeneous terms in the noncommutative parameter, we have added to the radial part the optional contribution  \be \frac{\lambda}{2} x_0\star \phi \label{addition}.\ee
With this choice,  and
 $\alpha/\beta=-1/4$,  we obtain a term proportional to the ordinary Laplacian, multiplied by $x_0^2$, plus lower derivatives.

The term   \eqn{addition} is not  relevant for our subsequent analysis, therefore we will  ignore it in the rest of the paper, as it only produces a shift in the spectrum of the radial operator from $j^2$ to $j(j+1/2)$.  Nor it is really relevant for the homogeneity in $\lambda$  of the various terms of the Laplacian coming form \eqn{radialact}: we shall see  below, on expanding the noncommutative field $\phi$ in the matrix base, that the different order in $\lambda$ of the various terms in \eqn{radialact} is only  a fictitious one, which does not take into account the dependence on the noncommutative parameter of the field itself . Indeed, as will be clear   from  Eqs. \eqn{x+x-}-\eqn{x0x0},  the whole term $\Delta \phi$ is of order $\lambda^2$ . For simplicity, in the rest of the paper we restrict the analysis to $\alpha,\beta$ positive, which is a sufficient condition for the spectrum to be positive.

A rigorous analysis of the commutative limit should be performed  in terms of observables and correlation functions. We have not addressed this issue in the present work and plan to study it elsewhere, in connection with the problem of introducing a Laplacian operator without the rescaling factor $x_0^2$. This is an interesting point; our Laplacian is a natural one for $\R^3_\lambda$: it is constructed
in terms of derivations of the algebra supplemented by  multiplicative operators. We signal the reference \cite{presnajder} where a different Laplacian is proposed  for $\R^3_\lambda$ in the context of noncommutative quantum mechanics, to study the hydrogen atom. It would be interesting to apply such proposal to QFT. However, that operator is not based on derivations of $\R^3_\lambda$. There
might be other candidates; this issue is under investigation. 

To rewrite the action in the matrix base $\{v^j_{m\tilde m}\}$ we first express the coordinate functions in such a base.
On using the expression of the generators in terms of $\bar z_a, z_a$, \eqn{xmu} and the base transformations \eqn{changeofbase} we find
\beqa
x_+ &=&\frac{\lambda}{\theta}\bar z_1 z_2 = \lambda \sum_{j,m}\sqrt{(j+m)(j-m+1)} v^j_{m\,m-1}\\
x_- &=& \frac{\lambda}{\theta}\bar z_2 z_1 = \lambda \sum_{j,m}\sqrt{(j-m)(j+m+1)} v^j_{m\,m+1}\\
x_3 &=&\frac{\lambda}{2\theta} (\bar z_1 z_1-\bar z_2 z_2)=\lambda \sum_{j,m} m v^j_{m\,m}\\
x_0 &=&\frac{\lambda}{2\theta} (\bar z_1 z_1+\bar z_2 z_2)=\lambda \sum_{j,m} j v^j_{m\,m}
\eeqa
were we have introduced
\be
x_\pm:= x_1\pm i x_2.
\ee
Thus we compute
\be \label{x0v}
\begin{array}{lll}
x_+\star v^j_{m\tilde m}= \lambda \sqrt{(j+m+1)(j-m)} v^j_{m+1 \, \tilde m}  & & v^j_{m\tilde m}\star x_+= \lambda\sqrt{(j-\tilde m+1)(j+\tilde m)} v^j_{m\, \tilde m -1} \\
x_-\star v^j_{m\tilde m}= \lambda \sqrt{(j-m+1)(j+m)} v^j_{m-1\,\tilde m}  &&v^j_{m\tilde m}\star x_-= \lambda\sqrt{(j+\tilde m+1)(j-\tilde m)} v^j_{m \,\tilde m +1} \\
x_3\star v^j_{m\tilde m}= \lambda\, m \, v^j_{m\tilde m}&&  v^j_{m\tilde m} \star x_3= \lambda \,\tilde m \, v^j_{m\tilde m}\\
x_0\star v^j_{m\tilde m}= \lambda\, j \, v^j_{m\tilde m}&&  v^j_{m\tilde m} \star x_0= \lambda\, j \,v^j_{m\tilde m}
\end{array}
 \ee
 which yield
\beqa
[x_+,[x_-, v^j_{m\tilde m}]_\star]_\star &=& \lambda^2 \Bigl\{ \Bigl( (j+m)(j-m+1)+ (j+\tilde m+1) (j-\tilde m) \Bigr) v^j_{m\tilde m} +  \nn\\
&-&\sqrt{  (j+m)(j-m+1)(j+\tilde m) ( j-\tilde m +1) } v^j_{m-1 \tilde m-1} +\nn\\
& -& \sqrt{  (j+m+1)(j-m)( j+\tilde m+1) ( j-\tilde m ) } v^j_{m+1 \tilde m+1}  \Bigr\} \label{x+x-}\\
{[}x_-,[x_+, v^j_{m\tilde m}]_\star]_\star &=& \lambda^2 \Bigl\{ \Bigl( (j+m+1)(j-m)+ (j+\tilde m) (j-\tilde m+1) \Bigr) v^j_{m\tilde m} +  \nn\\
&-&\sqrt{  (j+m)(j-m+1)( j+\tilde m) ( j-\tilde m+1 ) } v^j_{m-1 \tilde m-1}  +\nn\\
& -& \sqrt{  (j+m+1)(j-m)(j+\tilde m+1) ( j-\tilde m) } v^j_{m+1 \tilde m+1}     \Bigr\} \label{x-x+}\\
{[}x_3,[x_3, v^j_{m\tilde m}]_\star]_\star&=& \lambda^2 (m-\tilde m)^2 v^j_{m\tilde m}\nn\\
x_0\star x_0 \star v^j_{m\tilde m}&=& \lambda^2 j^2 v^j_{m\tilde m} \label{x0x0}
\eeqa
These relations allow to compute
\beqa
\Delta(\alpha,\beta) v^j_{m\tilde m}&=& \frac{\alpha}{\kappa^4} \Bigl( \frac{1}{2}( [x_+,[x_-, v^j_{m\tilde m}]_\star]_\star + [x_-,[x_+, v^j_{m\tilde m}]_\star]_\star )+ x_3,[x_3, v^j_{m\tilde m}]_\star]_\star \Bigr)\nn\\
&+&\frac{\beta}{\kappa^4} x_0\star x_0 \star v^j_{m\tilde m}.
\eeqa

On using the expansion of the fields in the matrix base and the multiplication rule for the base elements, already introduced in the previous section, respectively
\be
\phi=\sum_{j,m\tilde m} \phi^j_{m\tilde m} v^j_{m \tilde m} \label{phican}
\ee
and
\be
v^j_{m\tilde m}\star v^{\tilde\jmath}_{n\tilde n} = \delta^{j\tilde\jmath} \delta_{\tilde m n}v^j_{m\tilde n}
\ee
we obtain the
action  in \eqn{action} as a matrix model action
\beqa
S[\phi]&=&  \kappa^3\big\{\sum \phi^{j_1}_{m_1\tilde m_1}\bigl(\Delta(\alpha,\beta)+\mu^2\mathbf{1}\bigr)^{j_1 j_2}_{m_1 \tilde m_1; m_2\tilde m_2} \phi^{j_2}_{m_2\tilde m_2} + {{g}\over{4!}}\sum \phi^{j_1}_{mn} \phi^{j_2}_{np} \phi^{j_3}_{pq} \phi^{j_4}_{qm} \delta_{j_1 j_2} \delta_{j_2 j_3} \delta_{j_3 j_4}\bigr\}\nn\\
&=&
\kappa^3\big\{\Tr(\Phi(\Delta(\alpha,\beta)+\mu^2\mathbf{1})\Phi)+{{g}\over{4!}}\Tr(\Phi\Phi\Phi\Phi)\big\}
\label{action1}
\eeqa
where sums are understood over all the indices and $\Tr:=\sum_j \Tr_j$.  The matrix elements of the identity operator are
\be
\mathbf{1}^{j_1 j_2}_{m_1\tilde m_1 m_2\tilde m_2}=  \delta^{j_1 j_2} \delta_{\tilde m_1 m_2} \delta_{m_1\tilde m_2}
\ee
The kinetic operator may be verified  to be
\beqa
(\Delta(\alpha,\beta)+\mu^2\mathbf{1})^{j_1 j_2}_{m_1\tilde m_1;m_2\tilde m_2}&:=& \frac{1}{\pi^2\theta^2} \int v^{j_1}_{m_1 \tilde m_1}\star( \Delta (\alpha,\beta)+\mu^2\mathbf{1}) v^{j_2}_{m_2\tilde m_2}
\nn\\
&=& \frac{\lambda^2}{\kappa^4}\delta^{j_1 j_2}\bigl\{
\delta_{\tilde m_1 m_2}\delta_{m_1 \tilde m_2} D^{j_2}_{m_2\tilde m_2}-\delta_{\tilde m_1, m_2+1}\delta_{m_1, \tilde m_2+1}B^{j_2}_{m_2,\tilde m_2} \nn\\
&-&\delta_{\tilde m_1, m_2-1}\delta_{m_1, \tilde m_2-1} H^{j_2}_{m_2,\tilde m_2}\,\bigr\}
\label{kineticterm}
\eeqa
with
\beqa
D^j_{m_2\tilde m_2}&=&[({2\alpha}+{\beta})j^2 +  {2\alpha}(j_2- m_2 \tilde m_2) ] +\lambda^2{\mu^2} \\
B^j_{m_2\tilde m_2}&=& \alpha \sqrt{(j+m_2+1)(j-m_2)(j+\tilde m_2+1)(j-\tilde m_2)}\\
H^j_{m_2\tilde m_2}&=&  \alpha\sqrt{(j+m_2)(j-m_2+1)(j+\tilde m_2)(j-\tilde m_2+1)}.
\eeqa
At this stage, some comments are in order.
\begin{romanlist}
\item The use of the matrix base for $\R^3_\lambda$ yields an interaction term which is diagonal (i.e a simple trace of product of matrices built from the coefficients of the fields expansion) whereas the kinetic term is not diagonal. Had we used the expansion of $\phi$ in the so called fuzzy harmonics base $(Y^j_{lk})$, $j\in{{\mathbb{N} }\over{2}}$, $l\in\mathbb{N}$, $0\le l\le 2j$, $-l\le k\le l$ (see below),  then we would have obtained a diagonal kinetic term with a complicated interaction term. Notice that this remark holds for any polynomial interaction term. We will come back to this point in a while.
 \item We observe that the action \eqref{action1} is expressed as an infinite sum of contributions, namely $S[\Phi]=\sum_{ j\in{{\mathbb{N} }\over{2}}}S^{(j)}[\Phi]$,
where the  expression for $S^{(j)}$ can be read off from \eqref{action1} and describes a scalar action on the fuzzy sphere $\mathbb{S}^j\simeq End({\cal{V}}^j)$.
 \end{romanlist}
A lot of information about the short and long distance behaviour of a matrix model with diagonal interaction term, regarding renormalization properties, is encoded
into the propagator. The computation of this latter amounts to the determination of the inverse of the kinetic operator in Eq. \eqn{action1} operator  which, because of  the remark ii) above, is  expressible into a block diagonal form. Explicitly
\begin{equation}
S_{Kin}[\Phi]=\kappa^3\sum_{j} \sum_{m,\tilde m} \phi^{j_1}_{m_1 \tilde m_1}(\Delta+\mu^2\mathbf{1})^{j_1 j_2} _{m_1\tilde m_1;m_2\tilde m_2}\phi^{j_2}_{m_2\tilde m_2}\label{kinetic1},
\end{equation}
with  $\Delta^{j_1 j_2} _{m_1\tilde m_1;m_2\tilde m_2}$ defined in \eqn{kineticterm}. Since the mass term is diagonal, let us put it to zero for the moment. We shall restore it at the end.
 One has the following ''law of indices conservation''
\begin{equation}
\Delta^{j_1 j_2}_{mn;kl}\ne 0 \implies
    j_1=j_2,\;\;\;  m+k= n+l\label{consindice1}
\end{equation}
We denote by $P^{j_1 j_2}_{mn;kl}(\alpha,\beta)$ the inverse of $\Delta^{j_1 j_2}_{mn;kl}(\alpha,\beta)$ which is defined by
\begin{equation}
\sum_{k,l=-j_2}^{j_2}\Delta^{j_1 j_2}_{mn;lk}P^{j_2 j_3}_{lk;rs}=\delta^{j_1 j_3}\delta_{ms}\delta_{nr},\ \sum_{m,n=-j_2}^{j_2}P^{j_1 j_2}_{rs;mn}\Delta^{j_2 j_3}_{mn;kl}=\delta^{j_1 j_3}\delta_{rl}\delta_{sk}\label{definvers},
\end{equation}
for which the law of indices conservation still holds true as
\begin{equation}
P^{j_1 j_2}_{mn;kl}\ne0\implies     j_1=j_2\;\;\;m+k= n+l\label{consindice2}.
\end{equation}
To determine $P^{j_1 j_2}_{mn;kl}$ one has to diagonalize $\Delta^{j_1 j_2}_{mn;kl}$. This can be done by direct calculation,  as in  \cite{Grosse:2003aj} for the case of noncommutative scalar theories on the Moyal spaces $\mathbb{R}^{2n}_\theta$, by first  exploiting the implications of \eqref{consindice1}, \eqref{consindice2} to turn, at fixed $j=j_1=j_2$,  the multi-indices quantities $\Delta^j_{mn;kl}$ and $P^j_{mn;kl}$ into two sets (indexed by, say, $m-n=l-k$, see \eqref{consindice1}, \eqref{consindice2}) of matrices with two indices, then by looking for a set of unitary transformations diagonalising the set of matrices stemming from $\Delta^j_{mn;kl}(\alpha,\beta)$. Then, each of these unitary transformations is found to be a solution of a 3-terms recursive equation  defining a particular class of orthogonal polynomials. In the case of the scalar noncommutative field theories considered in \cite{Grosse:2003aj}, the above recursive equations are solved by a specific family of Meixner polynomials \cite{kk}. Having diagonalized the kinetic operator, the expression for the propagator follows.
In the present case, a similar direct calculation can be performed and gives rise after some calculations to a family of 3-term recursive equations, solved by a particular family of  polynomials, the so called dual Hahn polynomials \cite{kk}.
Details of this derivation are given in  Appendix.

However the whole computation of the propagator can be considerably simplified observing that, as already mentioned,  we already know   an alternative orthogonal base for $End({\cal{V}}^j)$ where the  kinetic part of the action can be shown  to be diagonal. These are the so called fuzzy spherical harmonics. We show in the appendix that the two methods are equivalent,  and we exhibit the proportionality relation between dual Hahn polynomials and the fuzzy spherical harmonics.

\subsection{The kinetic action in the fuzzy spherical harmonics base}\label{subsect4.2}

It is well known and largely exploited in the  literature on the fuzzy sphere, that $End({\cal{V}}^j)$
 is spanned by  the so called Fuzzy Spherical Harmonics Operators, or, up to normalization factors,  irreducible tensor operators.
 We shall indicate them as
 \be
 \hat Y^j_{lk}\in End({\cal{V}}^j),\;\; l\in\mathbb{N},\;\; \ 0\le l\le 2j,\;\; -l\le k\le l,
  \ee
  whereas the unhatted objects $Y^j_{lk}$ are their symbols and are sometimes referred to as fuzzy spherical harmonics with no other specification (notice however
that the  functional form  of the symbols does depend on the dequantization map that has been chosen).
Concerning  the definition and normalization of the fuzzy spherical harmonics operators, we use the following conventions \cite{Das}. We set
\be
 J_{\pm}=\frac{{\hat x}_{\pm}}{\lambda}.
 \ee
We have, for $l=m$,
 \begin{equation}
\hat Y^j_{ll}:=(-1)^l  \frac{ \sqrt{2j+1}}{ l !}
\frac {\sqrt{ (2l+1)! (2j-l)! } } { (2j+l+1)!}  (J_+)^l   \label{highestfuzzyharmonics}
\end{equation}
while the others  are defined recursively through the action of $J_-$
\be
\hat Y^j_{lk}:= [ (l+k+1) (l-k)]^{-\frac{1}{2}} [J_-,\hat Y^j_{l,k+1}],\;\;  \label{fuzzyharmonics}
\ee
and satisfy
\be
(\hat Y^j_{lk})^\dag=(-1)^{k-2j} \hat Y^j_{l,-k},\
\langle \hat Y^j_{l_1 k_1},\hat Y^j_{l_2 k_2}\rangle=\Tr((\hat Y^j_{l_1k_1})^\dag{ \hat Y}^j_{l_2k_2})=(2j+1)\delta_{l_1l_2}\delta_{k_1k_2} \label{relatfuzzyharm}.
\end{equation}
The symbols are defined through the dequantization map \eqn{wigner}
\be
Y^j_{lk}:=\langle z|\,\hat Y^j_{lk}\,|z\rangle. \label{fuzzyha}
\ee
From \eqn{fuzzyharmonics}, \eqn{relatfuzzyharm} and the Lie algebra relation $[J_+,J_-]= 2J_3$ it is straightforward to check the usual properties
\beqa
{[} J_{-},{\hat Y}^j_{lk} { ] }&=& \sqrt{ (l+k)(l-k+1) } { \hat Y}^{j}_{l \, k-1}  \\
{[} J_{+},{\hat Y}^j_{lk} { ] }&=& \sqrt{(l-k)(l+k+1)} { \hat Y}^{j}_{l \, k+1} \\
{[}J_3,{\hat Y}^j_{lk} { ] } &=& k\; {\hat Y}^{j}_{l  k}\\
{[} J_{i},{[} J_{i},{\hat Y}^j_{lk} {]} {]}&=& l(l+1) {\hat Y}^{j}_{l  k}
\eeqa
which imply for the symbols
\beqa
{[}x_-,Y^j_{lk}{ ] }_\star=\lambda  <z|{[} J_{-},{\hat Y}^j_{lk} { ] }|z>= \lambda \sqrt{(l+k)(l-k+1)}Y^{j}_{l \, k-1} \nn\\
{[}x_+,Y^j_{lk}{ ] }_\star= \lambda  <z| {[} J_{+},{\hat Y}^j_{lk} { ] } |z>= \lambda \sqrt{(l-k)(l+k+1)}  Y^{j}_{l \, k+1}\nn\\
{[}x_3,Y^j_{lk}{ ] }_\star=\lambda  <z|{[}J_3,{\hat Y}^j_{lk} { ] } |z>=\lambda \; k \; Y^{j}_{l  k}\label{xharm}
\eeqa
and in particular
\be
{[}x_i,{[}x_i, Y^j_{lk}{ ]}_\star {]}_\star=\lambda ^2 <z|{[} J_{i},{[} J_{i},{\hat Y}^j_{lk} {]} {]} |z>=\lambda^2 \; l(l+1) \; Y^{j}_{l  k}. \label{xxharm}
\ee
In order to evaluate the action of the full Laplacian \eqn{lapl} on the fuzzy spherical harmonics we need to compute $x_0\star Y^j_{lk}$. To this we express the fuzzy spherical harmonics  in the canonical base $ v^j_{m\tilde m}$
\begin{equation}
 Y^j_{lk}=\sum_{-j\le m,\tilde m\le j} ( Y^j_{lk})_{m \tilde m} v^j_{m \tilde m},\  \label{matrixelemfuzzydef}
\end{equation}
where the coefficients are given in terms of Clebsch-Gordan coefficients by \cite{Das}
\begin{equation}
(Y^j_{lk})_{m \tilde m}=\langle \hat v^j_{m \tilde m}|\hat Y^j_{lk}\rangle={\sqrt{2j+1}}(-1)^{j-\tilde m}\cg {j}{m}{j}{-\tilde m}{l}{k},\ -j\le m,\tilde m\le j \, ,\label{matrixelemfuzzy}
\end{equation}
\be
({Y^j_{lk}}^\dag)_{m \tilde m}=(-1)^{-2j} (Y^j_{lk})_{\tilde m m}. \label{fuzzydag}
\ee
On using  Eq. \eqn{matrixelemfuzzydef}, and the orthogonality relation of Clebsch-Gordan coefficients
\be
\sum_{m \tilde m} \cg {j}{m}{j}{\tilde m}{l_1}{k_1}
\cg {j}{m}{j}{\tilde m}{l_2}{k_2}=\delta_{l_1 l_2}\delta_{k_1k_2} \label{orthoclebsch}
\ee
together with the star product
\eqn{ortov}  it is straightforward to check that
\beqa
\int {Y^{j_1}}_{l_1 k_1}\star Y^{j_2}_{l_2 k_2}&=&\kappa^3 (-1)^{k_1+2j_1}(2j_1 +1) \delta^{j_1 j_2}\delta_{l_1 l_2} \delta_{-k_1 k_2}\\
\int {Y^{j_1}}_{l_1 k_1}^\dag \star Y^{j_2}_{l_2 k_2}&=&\kappa^3 (2j_1 +1)  \delta^{j_1 j_2}\delta_{l_1 l_2} \delta_{k_1 k_2}
\eeqa
in accordance with the second of relations \eqn{relatfuzzyharm}.
Eq. \eqn{matrixelemfuzzydef},  and last of Eqs. \eqn{x0v} imply
\be
x_0\star Y^j_{lk}= \sum_{-j\le m,\tilde m\le j} ( Y^j_{lk})_{m \tilde m} \;  x_0\star  v^j_{m \tilde m}= \lambda j Y^j_{lk}. \label{xoharm}
\ee
Thus, from the definition of the Laplacian \eqn{lapl}  and Eqs. \eqn{xxharm},  \eqn{xoharm},  we verify that in the fuzzy spherical harmonics base the whole kinetic term is diagonal,
\begin{equation}
\Delta (\alpha,\beta) Y^j_{lk}= \frac{\lambda^2}{\kappa^4}\left(\alpha l(l+1)+\beta j^2 \right)Y^j_{lk}\    \;\;\; j\in{{\mathbb{N} }\over{2}},\ 0\le l\le 2j,\ l\in\mathbb{N},\  -l\le k\le l  \label{finallaplacescale}
\end{equation}
with eigenvalues
\be
\frac{\lambda^2}{\kappa^4}\gamma(j,l;\alpha,\beta):=  \frac{\lambda^2}{\kappa^4}\big( \alpha l (l+1)+\beta j^2  \big) . \label{diagdelta}
\ee
Note that \eqref{finallaplacescale}, \eqref{diagdelta}, with our choice for the dimensionality of the parameters $\lambda$, $\kappa$,  single out a natural choice for the UV and IR regimes, which correspond respectively to large or small  values of $\gamma(j,l;\alpha,\beta)$.
We can expand the fields  $\phi \in \R^3_\lambda$ in the fuzzy harmonics base
\begin{equation}
\phi=\sum_{j\in { {{ {\mathbb{N}} }\over{2} }}}\sum_{l=0}^{2j}\sum_{k=-l}^l\varphi^j_{lk} Y^j_{lk},\label{expansharmonic}
\end{equation}
and comparing with their expression in the canonical base,  \eqn{phican},
we readily obtain
\begin{equation}
\phi^j_{m \tilde m}=\sum_{l=0}^{2j}\sum_{k=-l}^l ( Y^j_{l\kappa})_{m\tilde m}\varphi^j_{lk}=\sum_{l=0}^{2j}\sum_{k=-l}^l {\sqrt{2j+1}}(-1)^{j-m_2}\cg {j}{m}{j}{-\tilde m}{l}{k}\varphi^j_{lk}\label{coordinatchange2},
\end{equation}
which relates the propagating degree of freedom $\varphi^j_{lk}$, for which the kinetic term of the action is diagonal, to the interacting degree of freedom $\phi^j_{m \tilde m}$ for which the interaction term is diagonal.

From \eqn{expansharmonic} and the properties of the fuzzy harmonics we derive, for a Hermitian field $\phi$,
\begin{equation}
\phi^\dag=\phi\ \Rightarrow\ \varphi^j_{lk}=(-1)^{-k-2j}\varphi_{l,-k}^{j*},\;\;\; \;\; j\in{{\mathbb{N} }\over{2}},\ l\in\mathbb{N},\ 0\le l\le 2j,\ -l\le k\le l\label{hermiteanharmonic}.
\end{equation}
Therefore, upon restoring the mass term,  we can compute the kinetic action in the fuzzy harmonics base
\beqa
\int \phi\star (\Delta+\mu^2) \phi&= & \frac{\lambda^2} {\kappa^4}\sum \varphi^j_{l_1 k_1} \varphi^j_{l_2 k_2}\left( \gamma(j,l_2;\alpha,\beta)+\frac{\kappa^4}{\lambda^2} \mu^2\right) \int Y^j_{l_1 k_1}\star Y^j_{l_2 k_2}
\label{diagkin}\\
&=& \frac{\lambda^2} {\kappa^4}\sum \varphi^{j*}_{l_1 k_1} \varphi^{j} _{l_2 k_2}\left (\gamma(j,l_2;\alpha,\beta)+\frac{\kappa^4}{\lambda^2}\mu^2\right) \int{ Y^j}^\dag_{l_1 k_1}\star Y^j_{l_2 k_2}\nn\\
&=&
\frac{\lambda^2}{\kappa} \sum | \varphi^{j} _{l k}|^2 (2j +1)( \gamma(j,l;\alpha,\beta)+\frac{\kappa^4}{\lambda^2}\mu^2)
\eeqa
which is positive for $\alpha,\beta\ge0$.
We define for further convenience
\be
(\Delta_{diag})^{j_1 j_2}_{l_1 k_1 l_2 k_2}= \frac{1}{\lambda^3}\int Y^{j_1}_{l_1 k_1} \star \Delta(\alpha,\beta)  Y^{j_2}_{l_2 k_2} = \frac{1}{\lambda^2}   (-1)^{k_1+2j_1} (2j_1+1)  \gamma(j_1,l_1;\alpha,\beta) \delta^{j_1 j_2}\delta_{l_1 l_2} \delta_{-k_1 k_2}.  \label{deltadiag}
\ee

\subsection{The propagator}\label{subsect4.3}
We can now state the following Lemma
\begin{lem}
Let $\R^3_\lambda$ be the noncommutative algebra  defined in \eqref{sumfuzzy}, with canonical base $\{ v^j_{m\tilde m}\}$, $j\in { {{ {\mathbb{N}} }\over{2} }},\;\ -j\le m,\tilde m\le j$,    together with the fuzzy spherical harmonics base $\{ Y^j_{lk}\}$, $j\in { {{ {\mathbb{N}} }\over{2} }},\;\ 0\le l\le 2j,\ -j\le m\le j$. The inverse of the kinetic operator, in the canonical base,
\be
\bigl(\Delta(\alpha,\beta)+\mu^2 \mathbf{1} \bigr)^{j_1 j_2}_{m_1\tilde m_1;m_2\tilde m_2}=\frac{1}{\pi^2\theta^2} \int v^{j_1}_{m_1 \tilde m_1}\star \bigl(\Delta (\alpha,\beta)+\mu^2 \mathbf{1}\bigr) v^{j_2}_{m_2\tilde m_2}
\ee
is given by
\begin{equation}
(P(\alpha,\beta))^{j_1 j_2}_{p_1,\tilde p_1;p_2 \tilde p_2}=\delta^{j_1 j_2}\sum_{l=0}^{2j_1}\sum_{k=-l}^l\int_0^\infty dt\  e^{-t(2j_1+1)(\frac{\lambda^2}{\kappa^4}\gamma(j,_1 l;\alpha,\beta)+\mu^2)}({Y^{j_1}_{l k}}^\dag)_{p_1  \tilde p_1}(Y^{j_2}_{lk})_{p_2\tilde p_2}\label{propagator1},
\end{equation}
where $\gamma(j,l;\alpha,\beta)$, the eigenvalues of the Laplacian operator, have been given in \eqn{diagdelta}.
\end{lem}
\begin{proof}
It is based on the so called Schwinger parametrization. For each positive operator $A$
we can write
\be
\frac{1}{A}= \int_0^\infty  dt\, e^{-t A}
\ee
This applies to the matrix elements of the kinetic operator in the diagonal (propagation) base
\beqa
&&\Bigr[(\Delta_{diag}+\mu^2\mathbf{1} )^{-1} \Bigr]^{j_1 j_2}_{l_1k_1l_2 k_2} =\int dt \,e^{-t (\Delta_{diag}+\mu^2 \mathbf{1})^{j_1 j_2}_{l_1 k_1 l_2 k_2}}   \nn\\
&&\;\; ={(-1)^{k_1+2j_1} } \int dt e^{-t(2j_1+1) (\frac{\lambda^2}{\kappa^4}\gamma(j_1,l_2;\alpha\beta)+\mu^2) } \,\delta^{j_1 j_2}\,\delta_{l_1 l_2} \,\delta_{-k_1 k_2} \label{schw}
\eeqa
Let us perform a change of base from the diagonal to the interaction (canonical) base.
We have
\beqa
\int \phi \star \Delta \phi&=&\kappa^3 \sum \phi^{j_1}_{m_1 \tilde m_1} \Delta^{j_1 j_2}_{m_1 \tilde m_1 m_2 \tilde m_2} \phi^{j_2}_{m_2 \tilde m_2}\nn\\
&=&\kappa^3
\sum \varphi^{j_1}_{l_1 k_1}(Y^{j_1}_{l_1 k_1})_{m_1 \tilde m_1} \Delta^{j_1 j_2}_{m_1 \tilde m_1 m_2 \tilde m_2}\varphi^{j_2}_{l_2 k_2} (Y^{j_2}_{l_2 k_2})_{m_2 \tilde m_2}
\eeqa
By comparing with the expression in the diagonal base \eqn{diagkin}
we obtain
\be
(\Delta^{j_1 j_2}_{diag})_{l_1 k_1 l_2 k_2}= (Y^{j_1}_{l_1 k_1})_{m_1\tilde m_1} \Delta^{j_1 j_2}_{m_1\tilde m_1 m_2\tilde m_2}( Y^{j_2}_{l_2 k_2})_{m_2\tilde m_2} \label{proptoint}
\ee
with inverse transformation
\be
 \Delta^{j_1 j_2}_{m_1\tilde m_1 m_2\tilde m_2} =\frac{1}{(2 j_1+1)^2}(Y^{j_1}_{l_1 k_1})_{m_1\tilde m_1} (\Delta^{j_1 j_2}_{diag})_{l_1 k_1 l_2 k_2}( Y^{j_2}_{l_2 k_2})_{m_2\tilde
m_2}\label{deltadeltadiag}
 \ee
The (massless) propagator is then
\be
[\Delta^{j_1 j_2})^{-1}]_{m_1\tilde m_1 m_2\tilde m_2}=  (Y^{j_1}_{l_1 k_1})_{m_1\tilde m_1}[(\Delta_{diag}^{j_1 j_2})^{-1}  ]_{l_1 k_1 l_2 k_2}  (Y^{j_2}_{l_2 k_2})_{m_2\tilde m_2}
\ee
On replacing the expression for the diagonal inverse \eqn{schw} and on using the first of Eqs. \eqn{relatfuzzyharm}  we arrive at
\beqa
[\Delta+\mu^2 \mathbf{1})^{-1}]^{j_1 j_2}_{m_1\tilde m_1 m_2\tilde m_2} &=& {(-1)^{-k+2j_1}} \delta^{j_1 j_2} \sum_{l=0}^{2j_1} \sum_{k=-l}^l \int dt e^{-t (2j_1+1) \bigl(\frac{\lambda^2}{\kappa^4}\gamma(j_1,l; \alpha,\beta)+\mu^2\bigr)} (Y^{j_1}_{l -k})_{m_1\tilde m_1} (Y^{j_2}_{l k})_{m_2\tilde m_2} \nn\\
&=& {\delta}^{j_1 j_2}  \sum_{l=0}^{2 j_1} \sum_{k=-l}^l\int  dt e^{-t (2j_1+1) \bigl(\frac{\lambda^2}{\kappa^4}\gamma(j_1,l; \alpha,\beta)+\mu^2\bigr)} (Y^{j_1\dag}_{l k})_{m_1\tilde m_1} (Y^{j_2}_{l k})_{m_2\tilde m_2}
\eeqa
which completes the proof. The result can be verified directly, by using the orthogonality properties of the fuzzy harmonics.
\end{proof}

\section{One-loop calculations}\label{sect5}
Once we have established   the form of the propagator in the matrix base,
\be
P^{j_1 j_2}_{m_1\tilde m_1 m_2\tilde m_2}= \delta^{j_1 j_2}\sum_{l=0}^{2j} \sum_{k=-l}^l\int  dt e^{-t (2j_1+1) \bigl(\frac{\lambda^2}{\kappa^4}\gamma(j_1,l; \alpha,\beta)+\mu^2\bigr)} (Y^{j_1\dag}  _{l k})_{m_1\tilde m_1} (Y^{j_2}_{l k})_{m_2\tilde m_2}
\ee
and of the vertex
\be
V^{j_1 j_2 j_3 j_4}_{p_1 \tilde p_1; p_2\tilde p_2; p_3 \tilde p_3; p_4\tilde p_4 }= \frac{g}{4!} \delta^{j_1 j_2}\delta^{j_2 j_3}\delta^{ j_3 j_4}\delta_{\tilde p_1 p_2}  \delta_{\tilde p_2 p_3}
 \delta_{\tilde p_3 p_4}  \delta_{\tilde p_4 p_1} \label{vert}
\ee
the computation of Feynman graphs (ribbon graphs) of every order is  fairly easy: it is just a matter of gluing together the appropriate number of propagators,  which are represented by  a double-line, while contracting them with the diagonal vertex  (see Fig. \ref{propfig}).
\begin{figure}[htb]
\epsfxsize=5.0 in
\begin{align}
\parbox{20\unitlength}{\begin{picture}(22,22)
\put(-70,6){\epsfig{file=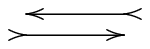,scale=0.9}}
      \put(-87,7){\mbox{\scriptsize$ j_1$}}
       \put(-79,-2){\mbox{\scriptsize$\tilde m_1$}}
       \put(-35,7){\mbox{\scriptsize$j_2$}}
       \put(-48,-2){\mbox{\scriptsize$m_2$}}
       \put(-79,17){\mbox{\scriptsize$m_1$}}
       \put(-48,17){\mbox{\scriptsize$\tilde m_2$}}
      \put(30,-24){\epsfig{scale=.9,file=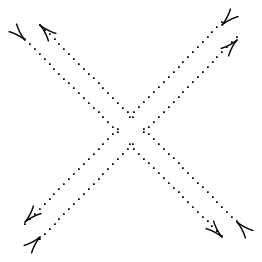}}
     \put(15,-24){\mbox{\scriptsize$j_1$}}
      \put(25,-13){\mbox{\scriptsize$p_1$}}
      \put(25,-34){\mbox{\scriptsize$\tilde p_1$}}
     \put(92,-24){\mbox{\scriptsize$j_2$}}
      \put(80,-34){\mbox{\scriptsize$p_2$}}
      \put(80,-13){\mbox{\scriptsize$\tilde p_2$}}
      \put(92,28){\mbox{\scriptsize$j_3$}}
      \put(80,39){\mbox{\scriptsize$\tilde p_3$}}
      \put(80,18){\mbox{\scriptsize$p_3$}}
     \put(15,28){\mbox{\scriptsize$j_4$}}
      \put(25,39){\mbox{\scriptsize$p_4$}}
      \put(25,18){\mbox{\scriptsize$\tilde p _4$}}
   \end{picture}}
\end{align}
\\
\caption{The propagator and the vertex.    }
\label{propfig}
\end{figure}
 Being independent on the details of the propagator, they can be obtained within a path-integral approach  as for example in \cite{beta},
 where the generating functional for connected correlation functions is explicitly computed  for  non-local matrix models (i.e. with non-diagonal propagator) with quartic interaction, up to second order in the coupling constant.
 In the following we explicitly compute typical  one-loop planar and non planar contributions  to the connected two and four point correlation functions.

\subsection{Planar two-point Green function}\label{planarsection}
A typical diagram contributing to the 2-point connected correlation function is depicted on Fig. \ref{planar}. Its amplitude is given by
\beqa
\mathcal{A}^{{j_1 j_2}^P}_{p_1\tilde p_1;p_2\tilde p_2}&=&\sum_{j_3 j_4=0}^{\infty} \sum_{p_3\tilde p_3 p_4\tilde p_4} V^{j_1 j_2 j_3 j_4}_{p_1 \tilde p_1; p_2\tilde p_2; p_3 \tilde p_3; p_4\tilde p_4 } P^{j_3 j_4}_{ p_4 \tilde p_4  p_3 \tilde p_3}=\delta^{j_1 j_2}\delta_{\tilde p_1 p_2} \sum_{ p_3}  P^{j_1 j_2}_{ p_1p_3 p_3 \tilde p_2}\nn\\
&=&\frac{\kappa^4}{\lambda^2} \delta^{j_1 j_2}\sum_{l=0}^{2j_1} \frac{1}{(2j_1+1) \bigl(\gamma(j_1,l,\alpha,\beta)+\frac{\kappa^4}{\lambda^2}\mu^2\bigr)}\sum_{ p_3}\sum_k (Y^{j_1\dag}_{lk})_{ p_1 p_3}
(Y^{j_2}_{lk})_{ p_3 \tilde p_2} \label{amplanarbase}
\eeqa
By using the expression of the fuzzy harmonics in the canonical matrix base, \eqn{matrixelemfuzzy}, \eqn{fuzzydag}, together with the relation
\be
\cg {j_1}{m_1}{j_2}{ m_2}{l}{k}=(-1)^{j_1-m_1}\sqrt{\frac{2l+1}{2j_2+1}}\cg {j_1}{m_1}{l}{-k}{j_2}{-m_2}
\ee
the sums over $ p_3$ and $k$ can be entirely performed. From the relation
$$
\frac{1}{2j+1}(Y^{j\dag}_{lk})_{ p_1 p_3}
(Y^{j}_{lk})_{ p_3 \tilde p_2}= (-1)^{-p_1-\tilde p_2} \cg {j}{ p_3}{j}{ -p_1}{l}{k} \cg {j}{ p_3}{j}{ -\tilde p_2}{l}{k}
$$
we obtain
\be
\mathcal{A}^{{j_1 j_2}^P}_{p_1\tilde p_1;p_2\tilde p_2}=\frac{\kappa^4}{\lambda^2}\delta^{j_1 j_2}\delta_{\tilde p_1 p_2}\delta_{p_1 \tilde p_2}\sum_{l=0}^{2j_1} (-1)^{2j_1}\frac{2l+1}{(2j_1+1)(\gamma(j_1,l;\alpha\beta)+\frac{\kappa^4}{\lambda^2}\mu^2)}\label{planar2pt}
\ee
It can be verified that ${\mathcal{A}^{j_1 j_2}}^P_{p_1\tilde p_1;p_2\tilde p_2}$ \eqref{planar2pt} is finite, including the case  $j\to\infty$. Indeed, let us pose $j_1=j_2=j$ and let us assume first that $\beta=0$. By a standard result of analysis, one can write
\begin{eqnarray}
\lim_{j\to\infty}|{\mathcal{A}^j}^{P(\beta=0)}_{p_1\tilde p_1;p_2\tilde p_2}|&=&\frac{\kappa^4}{\lambda^2}\delta_{\tilde p_1 p_2}\delta_{p_1 \tilde p_2}\lim_{j\to\infty}\frac{1}{2j+1}\sum_{l=0}^{2j}\frac{2l+1}{(\alpha l (l+1)+\frac{\kappa^4}{\lambda^2}\mu^2)}\nonumber
\\
&=&\frac{\kappa^4}{\lambda^2}\delta_{\tilde p_1 p_2}\delta_{p_1 \tilde p_2}\lim_{j\to\infty}\frac{1}{2j+1}
\int^{2j}_{0} dx\ {{2x+1}\over{(\alpha x(x+1)+\frac{\kappa^4}{\lambda^2}\mu^2) }}=0\label{converplanar},
\end{eqnarray}
showing finitude of ${\mathcal{A}^j}^{P(\beta=0)}_{p_1\tilde p_1;p_2\tilde p_2}$. This extends to ${\mathcal{A}^j}^P_{p_1\tilde p_1;p_2\tilde p_2}$ since $|{\mathcal{A}^j}^P_{p_1\tilde p_1;p_2\tilde p_2}|\le|{\mathcal{A}^j}^{P(\beta=0)}_{p_1\tilde p_1;p_2\tilde p_2}|$ holds true.\par

\begin{figure}[htb]
\epsfxsize=6.0 in
\begin{align}
\parbox{30mm}{\begin{picture}(22,22)
      \put(0,-30){\epsfig{scale=.9,file=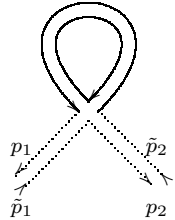}}
      \put(-1,-17){\mbox{\scriptsize$p_1$}}
      \put(-1,-40){\mbox{\scriptsize$\tilde p_1$}}
      \put(50,-40){\mbox{\scriptsize$p_2$}}
      \put(50,-17){\mbox{\scriptsize$\tilde p _2$}}
\end{picture}}\nn
\end{align}
\\
\caption{Planar diagram contributing to the 2-point correlation function}
\label{planar}
\end{figure}
In order to compare with known results on the fuzzy sphere we look for an expression of the planar amplitude in the fuzzy harmonics base. To this, we write the
 contribution of the planar two-point amplitude to the quadratic part of the effective action. At   one loop we have, up to multiplicative combinatorial factors (see \cite{beta})
\begin{eqnarray}
\Gamma^{(2) P}&=&\sum_{j_i, p_i,\tilde p_i}\phi^{j_1}_{p_1\tilde p_1}{\mathcal{A}^{j_1 j_2}}^P_{p_1\tilde p_1;p_2\tilde p_2}\phi^{j_2}_{p_2\tilde p_2}=\sum_{j_i, l_i, k_i}\varphi^{j_1}_{l_1k_1}\big(\sum_{p_i,\tilde p_i,}(Y^{j_1}_{l_1k_1})_{p_1\tilde p_1}{\mathcal{A}^{j_1 j_2}}^P_{p_1\tilde p_1;p_2\tilde p_2} (Y^{j_2}_{l_2k_2})_{p_2\tilde p_2}\big)\varphi^{j_2}_{l_2k_2}\nonumber
\\
&=&\sum_{j_i;l_i,k_i}\varphi^{j_1}_{l_1k_1}{\tilde{\mathcal{A}}}^{j_1 j_2\,P}_{l_1k_1;l_2k_2} \varphi^{j_2}_{l_2k_2}     \;\;\;\;\;\;\;\; i=1,2\label{changebasis}
\end{eqnarray}
where  we used $\phi^j_{mn}=\sum_{l,k}\varphi^j_{lk}(Y^j_{lk})_{mn}$ to obtain the  middle equality in \eqref{changebasis}, while the rightmost equality defines ${\tilde{\mathcal{A}}}^{{j_1 j_2}\,P}_{l_1k_1;l_2k_2}$, the amplitude in fuzzy harmonics  base. We have
\be
{\tilde{\mathcal{A}}}^{{j_1 j_2}\,P}_{l_1k_1;l_2k_2}= \frac{\kappa^4}{\lambda^2}\delta^{j_1 j_2}\sum_{p_i \tilde p_i=-j_i }^{j_i}\delta_{\tilde p_1 p_2}\delta_{p_1 \tilde p_2}\frac{ (-1)^{2j_1}}{2j_1+1}\sum_{l=0}^{2j_1}\frac{2l+1}{\alpha l(l+1)+\beta j_1^2+ \frac{\kappa^4}{\lambda^2}\mu^2} (Y^{j_1}_{l_1 k_1})_{ p_1 \tilde p_1} (Y^{j_2}_{l_2 k_2})_{p_2\tilde p_2}.
\ee
 The sum over $p_i, \tilde p_i$ can be performed thanks to the Kronecker delta symbols, giving rise to
\be
{\tilde{\mathcal{A}}}^{{j_1 j_2}\,P\ }_{l_1k_1;l_2k_2}= \frac{\kappa^4}{\lambda^2}\delta^{j_1 j_2}\sum_{l=0}^{2j_1} \frac{2l+1}{\alpha l(l+1)+ \beta j_1^2 + \frac{\kappa^4}{\lambda^2}\mu^2} (-1)^{k_2} \delta_{-k_1 k_2}\delta_{l_1l_2}.  \label{amplitplanarres}
\ee
 In order to establish a connection with the results which have been obtained in the literature on the fuzzy sphere \cite{vaidya01, chu01}, we fix $j_1=j_2=j$  and    $\beta=0$  so that our kinetic operator reproduces the Laplacian mostly considered within fuzzy sphere studies. We obtain
 \be
{\tilde{\mathcal{A}}}^{{j}\,P\ (\beta=0)}_{l_1k_1;l_2k_2}= \frac{\kappa^4}{\lambda^2}\sum_{l=0}^{2j} \frac{2l+1}{\alpha l(l+1)+\frac{\kappa^4}{\lambda^2}\mu^2} (-1)^{k_2} \delta_{-k_1 k_2}\delta_{l_1l_2}\label{amplitplanarres2}
\ee
which coincides with the result found in \cite{vaidya01,chu01}. \par
Let us notice that, unlike the result in the canonical base, Eq. \eqn{converplanar},  this amplitude  is logarithmically divergent  with $j\rightarrow \infty$.
Indeed, its behaviour is ruled by
the behaviour of the sum
\begin{equation}
{\cal{P}}=\sum_{l=0}^{n}{{2l+1}\over{l(l+1)+\nu^2 }},\;\;\;\; \nu^2:=\frac{\kappa^4}{\lambda^2}{{ \mu^2}\over{\alpha }}\;\;\;\;\; n=2j \label{serieplan2pt}
\end{equation}
We assume $\nu>0$. It can be readily observed that \eqref{serieplan2pt} is divergent for $n\rightarrow \infty$. Indeed,  let us introduce the following positive function on $[0,+\infty[$, \; $f(x)={{2x+1}\over{x(x+1)+\nu^2 }}$. One can check that it is monotonically decreasing on $[0,+\infty[$ provided $\nu^2\le{{1}\over{2}}$. Then, from an elementary result of analysis, $\lim_{n\to\infty}\sum_{l=0}^n{{2l+1}\over{l(l+1)+\nu^2 }}$ behaves as $\int_0^\infty dx\ f(x)$. But $\int_0^\infty dx\ f(x)=[\log(x(x+1)+\nu^2) ]_0^\infty$ which diverges and  so does \eqref{serieplan2pt}. When $\nu^2\ge{{1}\over{2}}$, the above analysis holds true provided one replaces  the domain of $f(x)$ by the domain on which $f(x)$ is decreasing and modifies accordingly the lowest value in the summation of the series.

The divergence developed in the propagating base is however only an apparent one, as can be seen inverting the relation \eqn{amplitplanarres} for the amplitude in the canonical matrix base. Let us consider this in detail. Eq. \eqn{matrixelemfuzzydef} and the orthogonality of Clebsch-Gordan  imply
\be
v^j_{m\tilde m}= \frac{(-1)^{2j}}{2j+1} \sum_{lk} (Y^j_{lk})_{m\tilde m} Y^j_{lk}
\ee
so that,
inverting Eq.  \eqn{changeofbase} we obtain
\be
\varphi_{lk}^j= \frac{(-1)^{2j}}{2j+1} \sum_{m \tilde m} \phi_{m\tilde m}^j (Y^j_{lk})_{m\tilde m} \label{changeofbaseinv}
\ee
Thus, from \eqn{changebasis}  we obtain the two-point planar amplitude in the canonical matrix base, in terms of the one in the fuzzy harmonics base
\be
\mathcal{A}^{j_1 j_2 \,P}_{m_1\tilde m_1;m_2\tilde m_2}=\frac{(-1)^{2(j_1+j_2)}}{(2j_1+1)(2j_2+1)} \sum_{l_i k_i} (Y^{j_1}_{l_1 k_1})_{m_1 \tilde m_1} \tilde{\mathcal{A}}^{j_1 j_2 \,P}_{l_1k_1;l_2k_2} (Y^{j_2}_{l_2 k_2})_{m_2 \tilde m_2}.
\ee
The latter being proportional to $(-1)^{k_2}$, the former becomes an alternating sum, which explains our results.

\subsection{Non-planar two-point graph}\label{nonplanarsection}
A typical non-planar contribution to the connected  2-point correlation function at one loop  is represented in Fig. \ref{nonplanar}.
\begin{figure}[htb]
\epsfxsize=5.0 in
\begin{align}
\parbox{40mm}{\begin{picture}(22,22)
    \put(0,-20){\epsfig{scale=.9,file=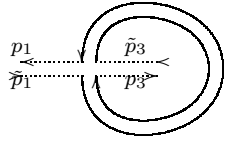}}
    \put(-3,0){\mbox{\scriptsize$\tilde p_1$}}
    \put(-3,12.5){\mbox{\scriptsize$p_1$}}
    \put(40,0){\mbox{\scriptsize$p_3$}}
    \put(40,12.5){\mbox{\scriptsize$\tilde p_3$}}
\end{picture}}\nn
\end{align}
\\
\caption{Nonplanar diagram contributing to the two-point function}
\label{nonplanar}
\end{figure}
The amplitude is given by
\beqa
{\mathcal{A}^{j_1 j_3}}^{NP}_{p_1\tilde p_1;p_3\tilde p_3}&=&\sum_{j_2 j_4\in\frac{\mathbb{N}}{2}} \sum_{p_2\tilde p_2 p_4\tilde p_4} V^{j_1 j_2 j_3 j_4}_{p_1 \tilde p_1; p_2\tilde p_2; p_3 \tilde p_3; p_4\tilde p_4 } P^j_{\tilde p_4 p_4 \tilde p_2 p_2}=\delta^{j_1 j_3}  P^{j_1 j_3}_{ p_1\tilde p_3  p_3 \tilde p_1} \nn\\
&=&\frac{\kappa^4}{\lambda^2}\delta^{j_1 j_3}\sum_{l=0}^{2j_1} \frac{1}{(2j_1+1) \bigl(\gamma(j_1,l,\alpha,\beta)+\frac{\kappa^4}{\lambda^2}\mu^2\bigr)}\sum_k (Y^{j_1\dag}_{lk})_{ p_1\tilde p_3}
(Y^{j_3}_{lk})_{ p_3 \tilde p_1}  \nn\\
&=&\frac{\kappa^4}{\lambda^2}\delta^{j_1 j_3}\sum_{l=0}^{2j_1}  \frac{1}{ \bigl(\gamma(j_1,l,\alpha,\beta)+\frac{\kappa^4}{\lambda^2}\mu^2\bigr)}\times \nn\\
&&\;\;\sum_k (-1)^{p_1+\tilde p_1}\cg{j_1}{\tilde p_3}{j_1}{-p_1}{l}{k}\cg{j_1}{p_3}{j_1}{-\tilde p_1}{l}{k}\  \label{Nonplanfin}
\eeqa
We first consider the simpler case with $p_1=\tilde p_1$, $p_3=\tilde p_3$ and assume for a while that the dimensionless parameter $m^2:=\frac{\kappa^4}{\lambda^2}\mu^2$ satisfies $m^2\ge1$. Then, the following estimate holds true
\beqa
&&\sum_{l=0}^{2j_1}\sum_{k=-l}^l\frac{1}{ \bigl(\alpha l(l+1)+\beta j_1^2+\frac{\kappa^4}{\lambda^2}\mu^2\bigr)}{\cg{j_1}{ p_3}{j_1}{-p_1}{l}{k}}    \cg{j_1}{p_3}{j_1}{- p_1}{l}{k} \nn\\
&& \;\; \; \le\sum_{l=0}^{2j_1}\sum_{k=-l}^l{\cg{j_1}{ p_3}{j_1}{-p_1}{l}{k}}\cg{j_1}{p_3}{j_1}{- p_1}{l}{k} =1 \label{nplanestimate1}
\eeqa
so that  $|{\cal{A}}^{j_1 j_3}_{p_1,p_1;p_3 p_3}|$ satisfies
\begin{equation}
|{\cal{A}}^{j_1 j_3}_{p_1, p_1; p_3  p_3}|\le\frac{\kappa^4}{\lambda^2}\delta^{j_1 j_3}\,\ \forall j_i\in\frac{\mathbb{N}}{2}, -j_i\le p_i \le j_i \label{geneboundnp}.
\end{equation}
From \eqref{Nonplanfin}, \eqref{nplanestimate1} and \eqref{geneboundnp}, one concludes that ${{\mathcal{A}}^j}^{NP}_{p_1p_1;p_3 p_3}$ is always finite for any value of the external indices. Relaxing now the above assumption of equality among the external indices, one notices  that the eigenvalues in the spectrum of the operator describing the propagator vanish at large $j$ and are of finite degeneracy, signaling a compact operator, hence bounded. From the very definition of the operator norm, the 2nd equality in \eqref{Nonplanfin} implies that the nonplanar amplitude is finite.\par

The relevant expression for the amplitude in the propagation basis can be computed in a way similar to the one of subsection \ref{planarsection}. We obtain
\beqa
&&{\tilde{\mathcal{A}}}^{j_1 j_2\,{NP}}_{l_1k_1;l_2 k_2}=\sum_{p_i,\tilde p_i}\mathcal{A}^{j_1 j_2\,{NP}}_{p_1\tilde p_1 p_2\tilde p_2}(Y^{j_1}_{l_1k_1})_{p_1\tilde p_1}  (Y^{j_2}_{l_2k_2})_{p_2\tilde p_2}=\frac{\kappa^4}{\lambda^2}\delta^{j_1 j_2}\sum_{l=0}^{2j_1}  \frac{2j_1+1}{ \bigl(\alpha l(l+1)+\beta j_1^2+\frac{\kappa^4}{\lambda^2}\mu^2\bigr)}\times\nn\\
&& \;\; \sum_{k=-l}^l\sum_{p_i,\tilde p_i=-j_i}^{j_i} (-1)^k\cg{j_1}{\tilde p_2}{j_2}{-p_1}{l}{k}\cg{j}{p_2}{j}{-\tilde p_1}{l}{k} \cg{j_1}{ p_1}{j_1}{-\tilde p_1}{l_1}{k_1}\cg{j_1}{p_2}{j_1}{-\tilde p_2}{l_2}{k_2}\nn\\
&&\;\;
=\frac{\kappa^4}{\lambda^2}\delta^{j_1 j_2}\sum_{l=0}^{2j_1}\frac{(2j_1+1)(2l+1)}{ \bigl(\alpha l(l+1)+\beta j_1^2+\frac{\kappa^4}{\lambda^2}\mu^2\bigr)} (-1)^{l_1+l+2j_1-k_1}\delta_{l_1 l_2} \delta_{k_1,-k_2}\left\{\begin{array}{ccc} j_1&j_1&l_1\\ j_1&j_1&l \end{array} \right\} \label{npres}
\eeqa
where the last term is a Wigner 6j-symbol. To obtain the rightmost equation, we have used the relation between Clebsch-Gordan coefficients and Wigner 3j-symbols
\begin{equation}
\cg{j_1}{p_1}{j_2}{ p_2}{l}{k}={\sqrt{2l+1}}(-1)^{j_1-j_2+ k}{\begin{pmatrix} j_1&j_2&l \\ p_1&p_2&-k \end{pmatrix}}
\end{equation}
and the summation formula for the product of four Wigner 3j-symbols, as given for example in \cite{wolfr}. For  $\beta=0$, $j_1=j_2=j$  \eqn{npres} agrees with the expression found in \cite{vaidya01,chu01} for the fuzzy sphere. \par

Let us notice that the analysis of the IR behavior of the two-point non-planar graph we are considering is more complicated in the propagating base. This has already been considered in \cite{chu01} for the fuzzy sphere, and  their analysis extends trivially to our case. We shortly review it for completeness.
 Note first that when $l_1=0$, the Wigner 6j-symbol in \eqref{npres} can be simplified into
\begin{equation}
\left\{\begin{array}{ccc} j_1&j_1&0\\ j_1&j_1&l \end{array} \right\}={{1}\over{ 2j_1+1}}(-1)^{2j_1+l}
\end{equation}
which, combined with \eqref{npres} yields
\begin{equation}
{\tilde{\mathcal{A}}}^{{j_1,j_2,\ NP(\beta=0)}}_{l_1=0}:={\tilde{\mathcal{A}}}^{{j_1,j_2,\ NP(\beta=0)}}_{0,0;0,0}=
\frac{\kappa^4}{\lambda^2}\delta^{j_1 j_2}\sum_{l=0}^{2j_1}\frac{(2l+1)}{ \bigl(\alpha l(l+1)+ \frac{\kappa^4}
{\lambda^2}\mu^2\bigr)}={\tilde{\mathcal{A}}}^{{j_1 j_2 \,P\ (\beta=0)}}_{0,0;0,0}\label{comparplannonplan}
\end{equation}
where ${\tilde{\mathcal{A}}}^{{j_1 j_2 P\ (\beta=0)}}_{0,0;0,0}$ can be read off from \eqref{amplitplanarres}. Notice that this relation extends obviously to the case $\beta\ne0$. From \eqref{comparplannonplan}, one deduces that both planar and non-planar contributions for zero external momentum, $l_1=0$, have the same behavior. It is finite for finite $j$ while a logarithmic divergence appears at large $j$ since  $\lim_{j\to\infty}\sum_{l=l_0}^{2j}{{2l+1}\over{l(l+1)+\nu^2 }}\sim\lim_{j\to\infty}\int_{l_0}^{j} dx\ \frac{(2x+1)}{ \bigl(\alpha x(x+1)+ \frac{\kappa^4}
{\lambda^2}\mu^2\bigr)}\sim \lim_{j\to\infty}\log(j)$.\par

A rigorous analysis of the general case $l_1\ll j$ would require to make use of asymptotic for the Wigner 6j-symbols. Nevertheless, a reliable approximation can be obtained by using the Racah approximation for the Wigner 6j-symbols coefficients, as already used in \cite{chu01}
\begin{equation}
\left\{\begin{array}{ccc} j&j&l_1\\ j&j&l \end{array} \right\}\approx{{(-1)^{l_1+l+2j} }\over{2j}}P_{l_1}(1-
{{ l^2}\over{ 2j^2}})
\end{equation}
where $P_l(x)$ denotes the Legendre polynomial of order $l$. This approximation is accurate provided $l_1\ll j$ and $j\gg 1$. This yields
\begin{eqnarray}
{\tilde{\mathcal{A}}}^{{j,NP,\ (\beta=0)}}_{l_1k_1;l_2 k_2}&-&{\tilde{\mathcal{A}}}^{{j,P \ (\beta=0)}}_{l_1k_1;l_2 k_2}\nonumber\\
&\approx&
(-1)^{k_1}\delta_{k_1,k_2}\delta_{l_1,l_2}\frac{\kappa^4}{\lambda^2}\sum_{l=0}^{2j}\frac{(2l+1)}{ \bigl(\alpha l(l
+1)+ \frac{\kappa^4}{\lambda^2}\mu^2\bigr)}({{2j+1 }\over{ 2j}}P_{l_1}(1-{{ l^2}\over{ 2j^2}}) -1)\nonumber\\
&\approx& (-1)^{k_1}\delta_{k_1,k_2}\delta_{l_1,l_2}\frac{\kappa^4}{\lambda^2}\sum_{l=0}^{2j}\frac{(2l+1)}
{ \bigl(\alpha l(l+1)+ \frac{\kappa^4}{\lambda^2}\mu^2\bigr)}(P_{l_1}(1-{{ l^2}\over{ 2j^2}}) -1).
\label{approxmixing1}
\end{eqnarray}
By further assuming that $\frac{\kappa^2\mu}{\lambda}\ll j$, the sum in \eqref{approxmixing1} can be approximated by ($\varepsilon:={{1}\over{j}}$)
\begin{eqnarray}
\sum_{l=0}^{2j}\frac{(2l+1)}{ \bigl(\alpha l(l+1)+\frac{\kappa^4}{\lambda^2} \mu^2\bigr)}(P_{l_1}(1-
{{ l^2}\over{ 2j^2}}) -1)&\approx&{{1}\over{\alpha j}}\int_0^2du\ {{2u+\varepsilon }\over{u(u+\varepsilon)+
{{ m^2}\over{j^2 }} }}\big(P_{l_1}(1-{{ u^2}\over{2}})-1 \big)\nonumber\\
&\approx&{{1}\over{2\alpha j}}\int_{-1}^1{{dx}\over{1-x}}(P_{l_1}(x)-1)\nonumber\\
&=&-{{1}\over{\alpha j}}h(l_1)\label{aproxmixing2}
\end{eqnarray}
where $h(n):=\sum_{k=1}^nk^{-1}$, $h(0)=0$ denotes the harmonic number. Eqn. \eqref{aproxmixing2}  can be interpreted as the counterpart of the apparent logarithmic divergence appearing in the planar amplitude in the harmonic base. .

\subsection{A class of finite scalars models at $\alpha=0$}\label{finitemodels}
In this subsection, we set $\alpha=0$, $\beta\ne0$ in \eqref{lapl} and assume $\rho^2:=\frac{ \kappa^4\mu^2}{\lambda^2\beta }>0$. Then, the propagator \eqref{propagator1} simplifies into
\begin{equation}
P(\beta)^{j_1 j_2}_{p_1,\tilde p_1;p_2 \tilde p_2}:=(P(0,\beta))^{j_1 j_2}_{p_1,\tilde p_1;p_2 \tilde p_2}=\frac{ \kappa^4}{\lambda^2\beta } \;{{(-1)^{2j_1}}\over{j_1^2+\rho^2 }}\delta^{j_1 j_2}\delta_{\tilde p_1p_2 }\delta_{\tilde p_2p_1 },\ \forall j_i\in\frac{\mathbb{N}} {2} \label{propsimple}
\end{equation}
which can be easily obtained by using the properties of the fuzzy harmonics. This expression simplifies the computation of the amplitude of any diagram of arbitrary order since each sum over internal indices simply results in an overall factor $2j+1$. This will accordingly simplify the power counting analysis.\par

In order to prepare the ensuing discussion, we first consider the amplitude of the planar diagram for the 2-point function. From \eqref{amplanarbase} and \eqref{propsimple}, we find that the corresponding amplitude can be cast into the form
\begin{eqnarray}
\mathcal{A}^{{j_1 j_2}^P}_{p_1\tilde p_1;p_2\tilde p_2}(\alpha=0)&=&\delta^{j_1 j_2}\delta_{\tilde p_1 p_2}\sum_{\tilde p_3}  P(\beta)^{j_1}_{ p_1\tilde p_3\tilde p_3 \tilde p_2}={{ \kappa^4}\over{ \lambda^2\beta}}\delta^{j_1 j_2}\delta_{\tilde p_1 p_2}\delta_{p_1 \tilde p_2}(\sum_{\tilde p_3=-j_1}^{j_1}\delta_{\tilde p_3\tilde p_3}){{(-1)^{2j_1}}\over{j_1^2+\rho^2 }}\nonumber
\\
&=&{{\kappa^4}\over{ \lambda^2\beta}}\delta^{j_1 j_2}\delta_{\tilde p_1 p_2}\delta_{p_1 \tilde p_2}{{(-1)^{2j_1}(2j_1+1)}\over{j_1^2+\rho^2 }}\label{planarbeta0}.
\end{eqnarray}
In the same way, the amplitude for the nonplanar diagram is
\begin{equation}
\mathcal{A}^{{j_1 j_2}^{NP}}_{p_1\tilde p_1;p_2\tilde p_2}(\ \alpha=0)={{\kappa^4}\over{ \lambda^2\beta}}\delta^{j_1 j_2}\delta_{p_1 \tilde p_1} \delta_{p_2 \tilde p_2}{{(-1)^{2j_1}}\over{j_1^2+\rho^2 }}\label{nonplanarbeta0},
\end{equation}
which differs from \eqref{planarbeta0} by a factor $2j_1+1$ stemming from the sum over an internal index, that is an inner loop occurring in the planar amplitude. It can be readily verified that $\mathcal{A}_{p_1\tilde p_1;p_2\tilde p_2}^{{j_1 j_2}^{ P}}(\alpha=0)$ and ${{\mathcal{A}}^{j_1 j_2}}^{NP}_{p_1\tilde p_1;p_2\tilde p_2};(\alpha=0)$ are finite both for $j=0$ and $j\to\infty$ (recall $\rho^2\ne0$).\par

For this class of models the analysis of the degree of divergence may be carried out for a generic graph{\footnote{We are indebted to F. Vignes-Tourneret for this observation.}}.  Let ${\cal{A}}^j_{{\cal{D}}}$ denote the amplitude for an arbitrary ribbon diagram ${\cal{D}}$ with genus $g$ ($g$ is  the genus of the Riemann surface related to the diagram) and given $j$,  the momentum circulating in the diagram, which is conserved,. Recall that a ribbon graph is built from 2 lines (see e.g Fig.\ref{propfig}). The relevant topological properties of ${\cal{D}}$ are characterized (see e.g \cite{Grosse:2003aj}) by a set of integer numbers $(V,I,F,B)$ where $V$ and $I$ denote respectively the number of vertices and internal ribbon lines (counting the number of double lines propagators), $F$ denotes the number of faces (it can be determined simply by closing the external legs of ${\cal{D}}$ and counting the number of closed {\it{single}} lines) and the Euler characteristics $\chi$ of the related Riemann surface is
\begin{equation}
\chi:=2-2g=V-I+F\label{kig}
\end{equation}
Finally, $B$ is the number of boundary components which counts the number of closed lines having external legs. By noting that $F-B$ counts the number of internal summations, i.e inner loops, we can write (dropping the unessential overall constants)
\begin{equation}
|{\cal{A}}^j_{{\cal{D}}} |\le({{1}\over{j^2+\rho^2}})^{I}(2j+1)^{F-B}\label{upperboundbeta}.
\end{equation}
Since $\rho^2\ne0$, there is no singularity at $j=0$ while the finitude of ${\cal{A}}^j_{{\cal{D}}} $ \eqref{upperboundbeta} at $j\to\infty$ depends on the sign of
\begin{equation}
\omega({\cal{D}}):=2I+B-F=(I+B+V)+2g-2
\end{equation}
where \eqref{kig} has been used. This defines the power counting for the noncommutative scalar field theory at $\alpha=0$. Therefore the amplitude ${\cal{A}}^j_{{\cal{D}}} $ is finite provided
\begin{equation}
\omega({\cal{D}})\ge0,
\end{equation}
which holds true. This implies that the theory at $\alpha=0$ is finite.

\section{Discussion and conclusion}

Let us first summarize the main results of this paper. We have examined  a family of scalar NCFT on the noncommutative $\Rl$, a deformation of the Euclidean $\mathbb{R}^3$ through a noncommutative  associative product of Lie algebra type. We have constructed a natural matrix base adapted to $\Rl$. It involves the Wick-Voros symbols of the operators of the canonical base of $\bigoplus_{j\in { {{ {\mathbb{N}} }\over{2} }}}\mathbb{S}_j$. We have then considered a family of real-valued scalar actions with quartic interaction  on $\Rl$ whose  kinetic operator can be written as a linear combination of the square of the angular momentum and a part related to the Casimir operator of $\mathfrak{su(2)}$. Working in the natural matrix base, the action can be expressed as an infinite sum of scalar actions defined on the successive fuzzy spheres $\mathbb{S}_j$ that ``foliate'' the noncommutative space $\Rl$, with kinetic operator of Jacobi type. The computation of the propagator in this base, for which the interaction is diagonal, has been done and gives rise to a rather simple expression.

We have computed the planar and non-planar 1-loop contributions to the 2-point correlation function and examined their behavior. We find that they are finite for positive  $\alpha,\beta$. Moreover no singularities are found  in the external momenta (indices). This signals very likely the absence of UV/IR mixing that would destroy the perturbative renormalizability. In   the limit situation $\alpha=0$   we find that the resulting theory is finite to all orders in perturbation.

  From the dimensional properties of the kinetic operator (see remark below \eqref{diagdelta}), the region with low external indices
considered in the subsection \ref{nonplanarsection} corresponds naturally to the IR region (i.e low energy excitations from the kinetic operator).  As a conclusion, we do not expect that UV/IR mixing spoiling perturbative renormalizability shows up in the corresponding NCFT.

There are various potentially interesting directions which should be investigated. First, it is well known that  the commutative $\phi^4$ model is  super-renormalizable in three dimensions.  It would therefore be worthwhile to study  within the same scheme and with the same kinetic term as in this article, a model which is just renormalizable in the commutative framework, like the scalar $\phi^6$ model. Moreover, as we notice in the paper, the Laplacian we propose is proportional to the ordinary Laplacian times a factor of $x_0^2$ plus lower derivative terms. 
This is a natural one for $\R^3_\lambda$: it is constructed
in terms of derivations of the algebra, which are all inner,  supplemented by a  well defined  multiplicative operator.
A different proposal is suggested in \cite{presnajder} which is not based on derivations of the algebra.  This issue is presently under investigation, together with the analysis of the commutative limit.

The present analysis can be extended to the case of noncommutative gauge theories on $\Rl$ stemming from suitable versions of noncommutative differential calculus on $\Rl$. The resulting gauge-fixed actions have some features in common (but not all) with the scalar NCFT considered here. This is currently under study \cite{undertak1-2}.

 We remark finally that the associative product equipping $\Rl$ is rotationally but not translationally invariant. This, combined with the conclusion of this paper about the absence of dangerous UV/IR mixing seems to support the conjecture made in \cite{GLV09}, \cite{GLV08} relating translational invariance of the associative product to the possible occurrence of troublesome UV/IR mixing. This point must of course be clarified and deserves further investigation. \par

\setcounter{section}{0}

\appendix{Dual Hahn polynomials}

Dual Hahn polynomials are in the present framework  the counterpart of the Meixner polynomials at the root of the diagonalisation in \cite{Grosse:2003aj}.

Let us consider the expression of the kinetic action in the interaction base \eqn{kineticterm} which we report here for convenience
\beqa
(\Delta(\alpha,\beta)+\mu^2\mathbf{1})^j_{m_1\tilde m_1;m_2\tilde m_2}&=& \frac{1}{\lambda^2}\bigl\{
\delta_{\tilde m_1 m_2}\delta_{m_1 \tilde m_2} D^j_{m_2\tilde m_2}-\delta_{\tilde m_1, m_2+1}\delta_{m_1, \tilde m_2+1}B^j_{m_2,\tilde m_2} \nn\\
&-&\delta_{\tilde m_1, m_2-1}\delta_{m_1, \tilde m_2-1} H^j_{m_2,\tilde m_2}\,\bigr\} \label{kiac}
\eeqa
with
\beqa
D^j_{m_2\tilde m_2}&=&[({2\alpha}+{\beta})j^2 +  {2\alpha}(j- m_2 \tilde m_2) ] +\lambda^2{\mu^2} \\
B^j_{m_2\tilde m_2}&=& \alpha \sqrt{(j+m_2+1)(j-m_2)(j+\tilde m_2+1)(j-\tilde m_2)}\\
H^j_{m_2\tilde m_2}&=&  \alpha\sqrt{(j+m_2)(j-m_2+1)(j+\tilde m_2)(j-\tilde m_2+1)}.
\eeqa
To find the inverse of \eqn{kiac}
we go back to the $\R^4$ notation $v^j_{m_i\tilde m_i} \rightarrow v^{j}_{p_i \tilde p_i}$ with  $j+m_i=p_i$, $j+\tilde m_i=\tilde p_i$,  $j-m_i=q_i$, $j-\tilde m_i=\tilde q_i$ and  $i=1,2, p_i\ge 0$ (so that the indices are all  positive numbers).
The kinetic operator is then represented as
\beqa
(\Delta(\alpha,\beta)+\mu^2\mathbf{1})^j_{p_1\tilde p_1;p_2\tilde p_2}&=& \bigl\{
\delta_{\tilde p_1 p_2}\delta_{p_1 \tilde p_2} D^j_{p_2\tilde p_2}-\delta_{\tilde p_1, p_2+1}\delta_{p_1, \tilde p_2+1}B^j_{p_2,\tilde p_2} \nn\\
&-&\delta_{\tilde p_1, p_2-1}\delta_{p_1, \tilde p_2-1} H^j_{p_2,\tilde p_2}\,\bigr\}
\eeqa
with
\beqa
D^j_{p_2\tilde p_2}&=&\beta j^2+2\alpha j(1+p_2+\tilde p_2) -2\alpha p_2\tilde p_2 +\lambda^2{\mu^2} \\
B^j_{p_2\tilde p_2}&=& \alpha \sqrt{(p_2+1)(2j-p_2)(j+\tilde p_2+1)(2j-\tilde p_2)}\\
H^j_{p_2\tilde p_2}&=&  \alpha\sqrt{p_2(2j-p_2+1)\tilde p_2(2j-\tilde p_2+1)}.
\eeqa

We look for orthogonal polynomials which diagonalize the kinetic operator. We pose  $ \tilde p_1- p_1 =p_2-\tilde p_2= k$
so that
\be
(\Delta(\alpha,\beta)+\mu^2\mathbf{1})^j_{p_1\tilde p_1;p_2\tilde p_2}=(\Delta(\alpha,\beta)+\mu^2\mathbf{1})^j_{p_1, p_1+k ;\tilde p_2+k, \tilde p_2}
\ee
and we look for $U^{(j,k)}_{m i}$ such that
\be
(\Delta(\alpha,\beta)+\mu^2\mathbf{1})^j_{p_1, p_1+k ;\tilde p_2+k, \tilde p_2}
= \sum_i U_{p_1 i}^{k} \bigl(v_i+\mu^2\bigr)U_{i \tilde p_2}^{k} \label{eigenval}
\ee
with $v_i+\mu^2$ the eigenvalues of the kinetic operator and
\be
\sum U_{m i}^{(j,k)} U_{il}^{(j,k)}= \delta_{ml}. \label{orthonorm}
\ee
On multiplying on the left by $U_{p p_1}^{(j,k)}$ and summing over $p_1$ we arrive at
\be
U^{(j,k)}_{p\tilde p_2}(v_p) (D_{\tilde p_2,k+\tilde p_2}-v_p-\mu^2)- U^{(j,k)}_{p\tilde p_2+1}(v_p) B_{\tilde p_2+1,k+\tilde p_2+1}-
U^{(j,k)}_{p\tilde p_2-1}(v_p) H_{\tilde p_2-1,k+\tilde p_2-1}=0 \label{prehahneq}
\ee
We redefine
\be
U^{(j,k)}_{p\tilde p_2}(v_p) =f (p,k,N) \sqrt{\frac{(N-p_2)!p_2!}{\tilde p_2!(N-\tilde p_2)!}} V^{(N,k)}_{\tilde p_2}(v_p) \label{normfact}
\ee
with $N=2 j$ and $f$ a normalization factor.
We arrive at
\be
(D_{\tilde p_2-}v_p-\mu^2) V_{\tilde p_2}(v_p)-\alpha(N-\tilde p_2) (p_2+1) V_{\tilde p_2+1}(v_p) -\alpha \tilde p_2(N-p_2+1) V_{\tilde p_2-1}(v_p)=0 \label{hahneq}
\ee
On introducing
\be
B_{\tilde p_2}=-\alpha(N-\tilde p_2) (p_2+1)\, \;\;\; C_{\tilde p_2}= -\alpha\,\tilde p_2\,(N-p_2+1)
\ee
we have
\be
B_{\tilde p_2}+ C_{\tilde p_2} = - D_{\tilde p_2} +\frac{\beta}{4!}N^2+\mu^2
\ee
On redefining
\be
\tilde v_p= v_p -\frac{\beta}{4!} N^2-\mu^2
\ee
we finally obtain
\be
\bigl(-(B_{\tilde p_2}+ C_{\tilde p_2} ) -\tilde v_p \bigr) V_{\tilde p_2}(v_p) +B_{\tilde p_2}  V_{\tilde p_2+1}(v_p)+C_{\tilde p_2} V_{\tilde p_2-1}(v_p)=0
\ee
This is the equation satisfied by  a special class of orthogonal polynomials, the so called  {\it dual Hahn polynomials} \cite{kk}
 \be
 V_{\tilde p_2}(v_p)\equiv R_{\tilde p_2}(\lambda(p);\gamma,\delta,N)
 \ee
  with the identification $\tilde v_p=\lambda(p)= p(p+\gamma+\delta+1),\, \gamma=k, \,\delta=-k$.

  The dual Hahn polynomials are  given in terms of    hypergeometric functions as
  \be
  R_n(\lambda(x);\gamma,\delta,N) =\, {_3F_2}( -n,-x,x+\gamma+\delta+1; \gamma+1,-N;1) \;\;\;\; 0\le n\le N
  \ee
which are in turn defined in terms of Pochhammer-symbols. When one of the parameters in the first argument of the hypergeometric series is equal to a negative integer $-n$ the series becomes a finite sum, which is our case.    We refer to \cite{kk} for more details.

We have therefore obtained that the dual Hahn polynomials are the orthogonal polynomials which diagonalize the kinetic part of the action for our scalar model on $\R^3_\lambda$.

Dual Hahn polynomials and  fuzzy harmonics are indeed proportional. They are actually well known in nuclear physics and quantum chemistry, were they are also referred to as "discrete spherical harmonics" (see for example \cite{aquilanti}).

To clarify the relationship between dual Hahn polynomials and fuzzy harmonics let us reconsider Eq. \eqn{deltadeltadiag} where we map the kinetic action from the interaction to the propagating base, and let us multiply it by  $(Y^{j}_{l k})_{m_1\tilde m_1}$. For the sake of clarity we ignore the mass term. Upon summing over $m_1, \tilde m_1$ we obtain
\beqa
\Delta^j_{m_1\tilde m_1 m_2\tilde m_2}(Y^{j}_{l k})_{m_1\tilde m_1} &=& \frac{1}{2j+1}(-1)^{2j}\delta_{l l_1}\delta_{k k_1} (\Delta^j_{diag})_{l_1 k_1 l_2 k_2}( Y^j_{l_2 k_2})_{m_2\tilde m_2} \nn\\
&=&\frac{1}{\lambda^2}     \gamma(j,l;\alpha,\beta) ( Y^j_{l k})_{\tilde m_2 m_2}
\eeqa
which, on inserting the explicit form of $\Delta^j_{m_1\tilde m_1 m_2\tilde m_2}$, and recalling the relation \eqn{matrixelemfuzzy}, can be verified to be the  standard recurrence relation for Clebsch-Gordan coefficients
\begin{align}
[(j-m_2)(j-\tilde m_2)(j+m_2+1)(j+\tilde m_2+1)]^{ {{1}\over{2}}}(Y^j_{lk})_{m_2+1,\tilde m_2+1}\nonumber\\
+[(j+m_2)(j+\tilde m_2)(j-m_2+1)(j-\tilde m_2+1)]^{ {{1}\over{2}}}(Y^j_{lk})_{m_2-1,\tilde m_2-1}\nonumber\\
=[2j(j+1)-l(l+1)-2m_2\tilde m_2](Y^j_{lk})_{m_2,\tilde m_2},\ \forall j\in  {{ {\mathbb{N}} }\over{2} },\ -j\le m_2,\tilde m_2\le j\label{recurharmonic}.
\end{align}
But, as we have seen in Eq. \eqn{prehahneq}, this is also, up to the redefinition \eqn{normfact}, the recurrence relation for dual Hahn polynomials. To our knowledge this result has first appeared in \cite{smirnov}.  The precise relation  depends on the normalization chosen for the Hahn polynomials. Up to a normalization function which depends on $N, l, k$ we have
\be
R_{n}(\lambda(l);k,-k,\delta,N)= (-1)^{j-\tilde m_2}  \sqrt{\frac{(j+\tilde m_2)!(j-\tilde m_2)!}{(j+m_2)!(j-m_2)!}} (Y^j_{lk})_{m_2\tilde m_2}
\ee
with $n=\tilde p_2= j+\tilde m_2$, $k= m_2-\tilde m_2$, $N=2j$.

\vskip 2 true cm
\noindent
{\bf{Acknowledgments}}: We thank M. Dubois-Violette, H. Grosse, F. Lizzi and F. Vignes-Tourneret  for discussions and constructive comments. J.-C. W. thanks D. Blaschke and D. Perrot for discussions. P.V. is grateful to V. Rivasseau for exchanges at various stages of this work and for kind hospitality at LPT. She also   acknowledges a grant from the European Science Foundation
under the research networking project Quantum Geometry and Quantum Gravity, and partial
support by GDRE GREFI GENCO.

\end{document}